\documentclass[a4paper,UKenglish,cleveref, autoref, thm-restate]{lipics-v2021}
\usepackage{algorithm,algorithmic}
\usepackage{graphicx}
\usepackage{url}
\usepackage{subcaption}
\usepackage[T1]{fontenc}
\usepackage[utf8]{inputenc}
\usepackage{tabularx,ragged2e,booktabs,caption}
\newcolumntype{C}[1]{>{\Centering}m{#1}}

\newtheorem{prop}{Proposition}

\title{Flashback: Enhancing Proposer-Builder Design with Future-Block Auctions in Proof-of-Stake Ethereum} %TODO Please add

\author{Yifan Mao}{Computer Science and Engineering, the Ohio State University, United States}{mao.360@osu.edu}{}{}
\author{Mengya Zhang}{Computer Science and Engineering, the Ohio State University, United States}{zhang.9407@osu.edu}{}{}
\author{Shaileshh Bojja Venkatakrishnan}{Computer Science and Engineering, the Ohio State University, United States}{bojjavenkatakrishnan.2@osu.edu}{}{}
\author{Zhiqiang Lin}{Computer Science and Engineering, the Ohio State University, United States}{zlin@cse.ohio-state.edu}{}{}

 %TODO mandatory; please add comma-separated list of keywords

\category{} %optional, e.g. invited paper

\relatedversion{} %optional, e.g. full version hosted on arXiv, HAL, or other respository/website
\begin{document}
\nolinenumbers
\maketitle

%TODO mandatory: add short abstract of the document
\begin{abstract}
Maximal extractable value (MEV)—in which block proposers unethically gain profits by manipulating the order in which transactions are included within a block—is a key challenge facing blockchains such as Ethereum today. 
Left unchecked, MEV can lead to a centralization of stake distribution thereby ultimately compromising the security of blockchain consensus. 
To preserve proposer decentralization (and hence security) of the blockchain, Ethereum has advocated for a proposer-builder separation (PBS) in which the functionality of transaction ordering is separated from proposers and assigned to separate entities called builders. 
Builders accept transaction bundles from searchers, who compete to find the most profitable bundles.
Builders then bid completed blocks to proposers, who accept the most profitable blocks for publication. 
The auction mechanisms used between searchers, builders and proposers are crucial to the overall health of the blockchain. 
In this paper, we consider PBS design in Ethereum as a game between searchers, builders and proposers. 
A key novelty in our design is the inclusion of future block proposers—as all proposers of an epoch are decided ahead of time in proof-of-stake (PoS) Ethereum—within the game model. 
Our analysis shows the existence of alternative auction mechanisms that result in a better (more profitable) equilibrium to players compared to state-of-the-art. 
Experimental evaluations based on synthetic and real-world data traces corroborate the analysis. 
Our results highlight that a rethinking of auction mechanism designs is necessary in PoS Ethereum to prevent disruption. 
\end{abstract}

\section{Introduction}
% \mz{test} \ZQ{test} \ym{test} \sbv{test}
Blockchains have revolutionized the concept of trustless, decentralized applications (dapps) operating globally over the Internet. 
Despite fluctuations in sentiment, blockchains have evolved into a trillion-dollar market, solidifying their position as one of the most impactful technologies of recent times.
Prominent blockchains like Ethereum, Solana, Cardano, among others, feature smart contract capabilities, offering a versatile and potent platform for developing dapps across diverse domains. 
These domains span healthcare, gaming, social networks, digital assets, and beyond.
One particularly crucial domain for dapp development is decentralized finance (DeFi), encompassing financial services such as exchanges, lending and borrowing platforms, stablecoins, insurance, prediction markets, and more, all implemented as smart contracts~\cite{defiusecases}. 
DeFi services execute automatically through transparently implemented smart contract programs, eliminating the need for involvement from any third-party financial service provider, and have gained significant popularity.
Currently, DeFi on Ethereum alone accounts for hundreds of billions of dollars in trade volume~\cite{defimarketcommentary}.

An unintended consequence of proliferation of DeFi services is an increase in transaction reordering attacks on the blockchain~\cite{daian2020flash}.
% The network's nodes publicly store a mempool containing all pending transactions yet to be processed. 
% An unintended consequence of proliferation of defi services is an increase in transaction reordering attacks on the blockchain~\cite{daian2020flash}.
Block proposers in a blockchain have freedom in ordering transactions as they like before publishing a block. 
Furthermore, the collection of pending transactions
% , known as a proposer's mempool, 
remains visible to the public eye.
By observing an impending DeFi transaction,
% in the mempool, 
an attacker can issue its own transaction with higher transaction fees thereby incentivizing a block proposer to confirm the attacking transaction before the victim transaction and gain unfair rewards. 
This practice of front-running is tightly regulated (illegal in many cases) in the traditional financial system. 
In blockchains, front-running attacks not only diminish the payouts for victim users but can seriously compromise the security of the consensus protocol due to the vast amounts of profits attackers gain in such attacks~\cite{heimbach2022sok}.

% Mitigating transaction reordering attacks in blockchains is challenging and is an active area of research. 
To mitigate the negative effects of transaction reordering attacks, in Ethereum the task of arranging transactions within a block is delegated to third-parties called block builders. 
Block proposers receive pre-arranged blocks from builders for publication, 
% relinquishing the responsibility of transaction ordering to the builders.
while end-users (also called searchers) submit transactions or transaction bundles directly to builder(s) of their choosing including a fee with each transaction or bundle. 
During a block publishing slot, a proposer receives block bids from various builders from which the proposer accepts the most profitable block. 
Fees in transactions are shared between the block proposer and the builder who submitted the block. 
Searchers analyze pending high-value transactions and create new transactions and/or transaction bundles to potentially target these high-value transactions for attacks.
Builders thus compete to sell the greatest number of blocks to block proposers. 
Investing in better computational hardware, improving trust relationship with transaction providers and subsidizing blocks are some of the strategies builders use to best compete~\cite{mevlandscape,builddomandseardep,mevgalaxy}.

Builders in Ethereum today only advertise their blocks to the proposer of the next upcoming slot.
This is a remnant of a past practice, as Ethereum used proof-of-work consensus  till fall 2022 in which miners for future blocks are not known ahead of time. 
Time in proof-of-stake Ethereum is divided into 12 second slots and 32 slot epochs, with one block being produced per slot. 
At the beginning of an epoch proposers for all 32 slots in the epoch are decided via a pseudo-random sampling algorithm~\cite{eth2book}. 
Unfortunately, even post-merge (when Ethereum transitioned to proof-of-stake) the bidding policies of builders have not significantly changed. 
In this paper we consider the problem of an efficient searcher-builder-proposer auction design when the proposers of not only the upcoming slot, but also future slots are known in advance to builders. 
When a proposer is selected for a time slot in the future, this knowledge is known only to the selected proposer at first.
If there is an adequate incentive for the proposer to reveal this information to a block builder, a rational proposer would choose to do so.

% The rationale behind our design revolves around the anticipation that 
A proposer assigned to a future slot may wish to `reserve' high-value transactions accumulated by a builder, particularly if the likelihood of receiving such transactions during its future slot is low. 
We formulate the auction design as a game between three players: users (searchers), builders and proposers. 
Under the model where future block proposers are known to builders, we present Flashback---a novel block building and auction mechanism---and show the existence of an equilibrium where (1) a Flashback builder receives a greater reward compared to a default builder that does not consider future block proposers; (2) block proposers receive a greater share of blocks from a Flashback builder compared to the default builder, and (3) end-users submit a greater number of transactions to a Flashback builder compared to the default builder. 

Flashback presents a basic auction mechanism design to highlight the importance of future block proposers in PBS equilibria. 
% In this paper, we establish basic bidding strategies to demonstrate the existence of equilibrium. 
Employing advanced online optimization strategies could potentially offer improved outcomes, leaving room for further exploration in future research.
Considering our findings, we argue that the current Ethereum proposer-builder ecosystem is not at equilibrium and improved builder policies are possible. 
Multi-block auctions in which blocks for consecutive slots are collectively auctioned are practiced by some builders, but they are applicable only to instances where the same user is elected as a proposer for multiple consecutive blocks~\cite{jensen2023multi}. 
In summary, the contributions of this paper are: 
\begin{itemize}
\item We rethink optimal builder bidding strategy in proof-of-stake Ethereum leveraging knowledge of proposers for future time slots.
\item Under a game theoretic model, we show an auction mechanism that at equilibrium outperforms today's mechanisms.
\item Simulations using real-world transaction data, show our proposed builder obtaining 20\% higher rewards compared to today's builders.  
\end{itemize}

\section{Background}
%%- proof-of-stake Ethereum 
%%- notion of epoch
%%- how validators are chosen in each epoch 
\subsection{Ethereum}

\newcommand{\bheading}[1]{{\vspace{2pt}\noindent{\textbf{#1}}}}
\newcommand{\iheading}[1]{{\vspace{2pt}\noindent{\textit{#1}}}} 

Ethereum~\cite{Buterin2013} is a decentralized, public blockchain platform that is currently the second-largest after Bitcoin~\cite{nakamoto2008bitcoin}. 
% In Ethereum, the Peer-to-Peer (P2P) network plays a crucial role in maintaining the blockchain and propagating information such as blocks.

\bheading{Proof-of-Stake (PoS).}
Ethereum has switched to Proof-of-Stake (PoS) since September 15, 2022 from its previous consensus mechanism, Proof-of-Work (PoW)~\cite{posdoc}.
In PoS Ethereum, time is organized into epochs 
comprising 32 slots. 
% each of which persists for a span of 12 seconds. 
Each slot has one block proposer. 
At the onset of each epoch, proposers are assigned for all 32 slots in the epoch through shuffling. 
% , and subsequently designated as either block proposers or attestators within committees for each slot. 
% Analogous to miners in the PoW era, these validators assume the role of mining blocks.
A key reason for assigning proposers in advance is that an advance notification 
facilitates proposers in participating in the correct peer-to-peer network subnets and enables them to prepare for assigned tasks, such as attestation. 
However, this necessity for advance notification also poses a potential security challenge, as it compromises the inherent unpredictability that is vital to Ethereum's consensus mechanism~\cite{unpredictability}. 
\subsection{MEV, Builders, and Proposers}
\label{sec:mevbuildersandproposers}

\bheading{MEV.} 
Maximal extractable value (MEV) refers to the profits unfairly gained by attackers by executing transaction reordering attacks~\cite{daian2020flash}.  
Token swaps, arbitrage and loan liquidations are some examples of transactions from which MEV can be extracted. 
E.g., in a token swap a user wishes to swap one ERC-20 token $X$ for another token $Y$ at a decentralized exchange for a competitive price. 
By observing the user's transaction in the mempool (set of outstanding transactions), an attacker can issue its own transaction purchasing token $Y$ that frontruns the user's transaction, and later sell those tokens for a profit. 
Similarly, in an arbitrage a user exploits price difference of an asset at different exchanges and performs a sequence of buy-sell operations to gain a profit. 
However, an adversary observing the user's transaction in the mempool can copy the arbitrage and issue its own transaction that frontruns the user's transaction, thus gaining the profits for itself. 
It is reported that the monthly MEV collected on lending platforms and decentralized exchanges exceeds \$100M~\cite{heimbach2022sok}. 
%Flash Boys~\cite{daian2020flash} introduces how Maximum Extractable Value (MEV) % \mz{now we should call it "Maximum Extractable Value", not "Miner Extractable Value"} 
% happens during block generation.

% An Ethereum transaction sends money to some addresses which includes the gas and fees to the players that help to process the transactions in Ethereum network.
% Gas is consumed when the transaction be processed by the validator and the gas cost is same to all the processing transaction in that slot.
% Fees are given to the users that participate in processing the new block.
% Validator is the end user to decide the block and also the participator that takes most fees from the processing transactions.
% Such profits that are extracted by the validaors are noted as MEV.

\bheading{Proposer-builder separation.}
To mitigate the negative effects of MEV on proposer centralization, Ethereum has favored a proposer-builder separation architecture~\cite{heimbach2023ethereum, wahrstatter2023time}. 
% , in which the role of ordering transactions in a block is done by entities called builders while blocks are published by proposers running PoS consensus. 
We explain the pipeline of how transactions are confirmed in Ethereum's proposer-builder architecture below. 

\begin{enumerate}
\item Users generate new transactions including transaction fees indicative of the priority they desire for their transactions (high fees ensure quicker confirmation of transactions). 
These transactions are then broadcast across the network and are publicly visible to all the nodes.
% The network keeps a public mempool of newly generated transactions and  transactions that haven't been processed in the previous slots.
% Thus, transactions broadcast over the network are publicly visible even before they are confirmed on the blockchain. 
We refer to such transactions as public transactions. 

\smallskip 
Searchers act as adversaries in the network.
They keep searching for profitable victim transactions in the mempool.
When a victim transaction is found, a searcher constructs a transaction bundle including its own transactions and the victim transaction in an appropriate order, and privately sends the bundle to a builder~\cite{wang2022impact, gupta2023centralizing}. 
Note that unlike public transactions, privately sent transactions are visible only to the recipient builder. 
Searchers are also  willing to pay more fees (obtained from their MEV profits) for faster inclusion in the blockchain.

\smallskip 
To avoid attacks from searchers, a user can also send its transactions directly to a builder through a private channel, including an appropriate amount of fees.
The private channel ensures the transactions are only available to the builder to whom they were sent and will only appear in that builder's generated candidate block.
\smallskip 

\begin{figure}[htp]
  \centering
    \includegraphics[width=0.7\linewidth]{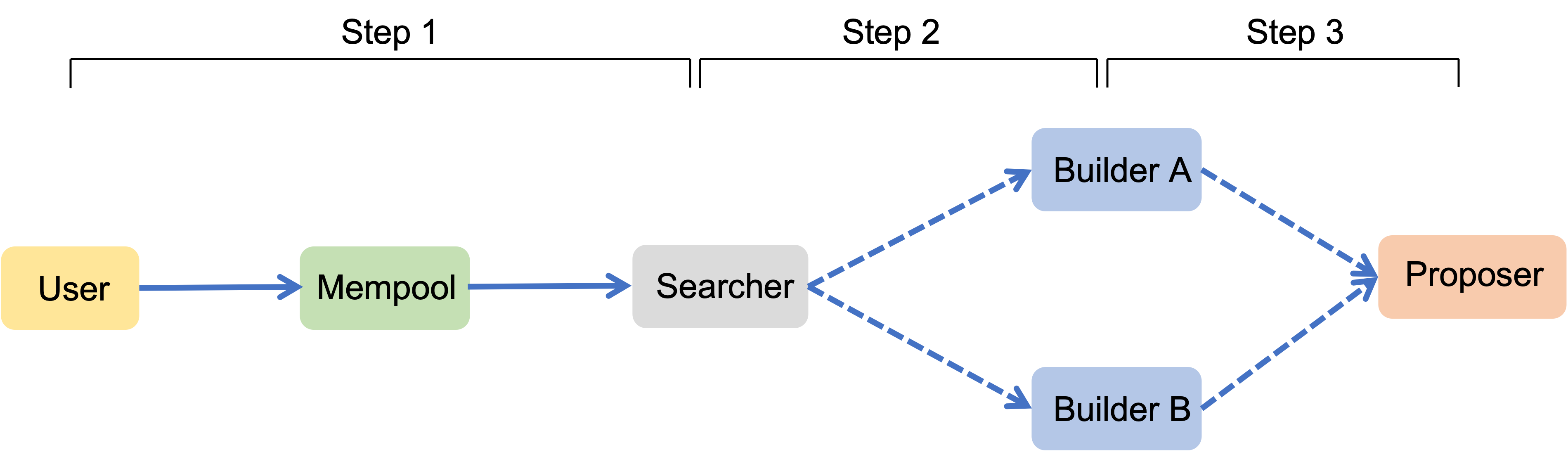}
  % \vspace{-1.2cm}
  \caption{Sequence of interactions between users, searchers, builders, and proposers in a PBS model for each round of block proposal (\S\ref{sec:mevbuildersandproposers}). Dashed edges denote private interactions. Private transactions submitted by users directly to the builders are not shown.} 
  \label{fig:background}
\end{figure}

\item 
Builders receive private transactions from users and transaction bundles from searchers, and package them into a candidate block. 
Part of the fees gained in the block is retained by the builder while most of it is marked for transfer to whichever proposer publishes the block. 
% including transaction order flow and fee for the block proposing validator.
The builder then advertises 
% sends the block to a relay which validates the transactions and checks whether the specified fees can be extracted.
% The relay then advertises 
the block to the proposer of the upcoming slot. 
% Thus, a relay helps a proposer to select a candidate block having the highest  fees.
% A relay works between builders and proposers to avoid power abusing at proposers.
% However, in our paper, 
% we consider the game on the rewards distributions and relay performs closely to builder to collect transactions adn generate valid candidiate block to the validator.
% In the later analysis and experiment sections, 
% we consider the builder and relay together as the builder.
%in the game.
\smallskip 

\item 
Proposers receive candidate blocks from builders, and choose the block generating the most profit for publication. 
Once the block is published, transactions within the block are deemed confirmed, and the proposer and the builder receive their share of fees as specified by the builder. 
Figure~\ref{fig:background} illustrates the process outlined so far. 
% can select the new block to process, which can include, exclude or reorder transactions in the new block.
% Due to the limited space in each transaction, only part of the transactions can be processed, and the order of processing affects the transaction's exchanging rates.
% Thus users are motivated to have their transactions to be processed as early as possible and the processing order inside the block as early as possible.
\end{enumerate}
% In PoS network, there are various roles participating generating a new block.
 % \ym{@mengya, mempool mentioned here}

% In the described actions, users have the option to pay fees through transaction fees and direct payments, both of which can be reclaimed within the transactions. 
% When a builder assembles a block, it consolidates the fees from these potential transactions and includes a payment to the validating validator in the concluding transaction of the block. 
% This arrangement benefits both the builder and the validator in creating new blocks. 
% The builder receives a portion of the fees for constructing the candidate block, which is ultimately validated by the validator. 
% The validator obtains the majority of the fees through the transfer within the block's final transaction.

\section{Model}
\label{s: model}

\subsection{System Model}
\label{s: netmodel}

\bheading{User model.}
We model our system as a game involving 3 types of entities: users, builders and proposers.
We do not distinguish between users and searchers in our model. 
Time is segmented into discrete rounds,  with a single block proposed by a designated proposer during each round.
We assume there are $n>0$ users.
We consider two builders in the system: the primary and secondary builder. 
The primary builder refers to a builder running Flashback auction, while the secondary builder models today's builders without future block auction capability. 
The words primary and secondary are used only for nomenclature and do not indicate any inherent user or validator preference of one builder over the other. 

Without loss of generality we assume at each round, an independent proposer publishes the block.  
Each round, each user generates a private transaction with a probability of $0<q<1$ (randomness independent across users).
When a user creates a private transaction, its transaction value is randomly sampled from a distribution $p_\mathrm{private}$. 
This value represents the total fees paid by the user for that specific transaction.

When a user generates a private transaction within a round, it randomly selects one of the two builders for submission. 
The probability of sending to the primary builder is $S_\mathrm{primary}/(S_\mathrm{primary} + S_\mathrm{secondary})$, while the probability of choosing the secondary builder is $S_\mathrm{secondary}/(S_\mathrm{primary} + S_\mathrm{secondary})$. 
Here, $S_\mathrm{primary}$ and $S_\mathrm{secondary}$ are scores indicating the performance of the primary and secondary builder respectively. 
The scores reflect various aspects such as the likelihood of a builder getting a block accepted based on past performance. 
The formal definition of the score function will be provided later.
A private transaction is assumed to expire if it remains unincorporated into the blockchain within $\tau>0$ rounds.

Additionally, each round involves the generation of $k$ public transactions, all of which are transmitted to both builders. 
The value of each generated public transaction is randomly sampled from the distribution $p_\mathrm{public}$. Public transactions are considered expired within the same round they are generated.

\bheading{Builder model.}
In each time slot, a builder receives private and public transactions from users, as described previously. 
Additionally, the builder maintains a repository of unexpired transactions that have yet to be confirmed. 
Alongside this, during each slot, the secondary builder (also referred to as the default builder) selects the most profitable $K$ unconfirmed transactions available to it. 
Here, we consider $K$ as the block size. Once the secondary builder has compiled a block for round $t$, it sends the block to the proposer assigned to that particular round. 
Upon successful inclusion of the block, the secondary builder receives a fraction $0< r_2 <1$ of the total transaction fees contained within the block. 
The remaining fraction, $1-r_2$, is received by the proposer who proposes this particular block.
It is important to note that during round $t$, the secondary builder interacts only with the proposer assigned to that round.

In each round, the primary builder, also known as the policy builder, employs a specific policy denoted as $\pi_\mathrm{build}$ to construct the block. 
During round $t$, the primary builder possesses information about identities of all validators from rounds $t$ to $t + c$, where $c > 0$ is a positive integer.
To advertise blocks and submit bids to proposers, the primary builder adheres to a bidding policy denoted as $\pi_\mathrm{bid}$. 
It is important to note that the fraction of transaction fees retained by the primary builder hinges on the interplay between the building policy $\pi_\mathrm{build}$ and the bidding policy $\pi_\mathrm{bid}$.

\bheading{Proposer model.}
% Each validator runs a policy $\pi_{\mathrm{validator}}$ for deciding how to accept blocks from the builders. 
We assume that proposers evaluate bids presented by both the primary and secondary builders. 
Each proposer's aim is to accept a block that maximizes its profits.
We use the words proposer and validator interchangeably. 

\bheading{Reward model.}
Users adhere to the policy described above, and thus, we do not assign any specific rewards to the users. 
For a builder, its reward is represented by the average fees obtained from the blocks it constructs that are confirmed. 
Similarly, a proposer's reward is determined by the amount of fees earned from the blocks it proposes. 

 \bheading{Score function.}
Users evaluate builders using a scoring function that considers several factors: fee charged by the builder, average waiting time for processing private transactions, and the rate of failure (expiry) of private transactions. 
% after extended waiting periods.
% This data is readily accessible through transaction records within the block.
Average waiting time signifies the duration from when transactions are received by the builder until their processing.
The failure rate refers to the proportion of timed-out transactions in a builder's private mempool that have expired.
% Historical information is publicly available in the blockchain, enabling all users to access and calculate these scores. 
While a user may know the above parameters for transactions it has generated so far, it may not know the parameters for transactions generated by other users. 
Therefore, we assume that users rate and share knowledge of their experience with different builders through public forums. 
% However, the initial time for the builder to receive private transactions and the outdated transaction rate remain secure and can only be accessed by users through community evaluations.
The overall score of a builder $b$ (for $b \in$ \{primary, secondary\}) is computed as 
% To standardize the three factors, we assign weights to each factor, normalizing the weights to align with user preferences when selecting builders: 
\begin{align}
S_b = w_r*F_r + w_d*F_d + w_m*F_m,
\end{align}
% \sbv{mathematically define what is $F_r, F_d, F_m$}
where 
\begin{enumerate}
\item $F_r$ is a moving average of the total reward earned by the builder, averaged over the most recent $W$ (we use $W=3200$ in our experiments) blocks. 
Published blocks are public and their contents accessible to all nodes in the network. 
Out of the fees paid by a user to a builder, typically the builder retains a small portion of the fees while the bulk of the fees is allocated to the proposer that publishes the block. 
The amount of fees a builder and proposer earned in a block can be readily computed from the block contents~\cite{etherscan}. 
% The last transaction in a block is a payment from the block's builder to the block's proposer, containing all the fees collected from users and transferred to the proposer.
% Users can easily determine the total transaction cost in the block, paid to the validator after deducting burnt fees. 
% They calculate a score, $F_r$, based on the average cost.
\item $F_d$ is a moving average of the time-until-expiry of transactions submitted to the builder experience before getting included in a published block.
Only unexpired transactions are considered in the computation of $F_d$.
% built by the builder. 
\item $F_m$ is a moving average of the number of transactions submitted to the builder which expired and failed to get included in a block.   
\item $w_r > 0, w_d > 0$ and $w_m < 0$ are weighting factors. 
A higher score for the builder implies users are more likely to choose the builder for sending their private transactions. 
\end{enumerate}
% \begin{align}
% F_r = Ave(Cost_{T_b}),
% \\
% F_d = Ave(Delay_{T_b}),
% \\
% F_m = Count(Missing_{T_b}),
% \end{align}
% for all the historical transactions $T_b$ processed by builder $b$.

In general, the processing delay experienced by a transaction, or its failure status, is not public knowledge and may not be available to all users. 
% requires private information concerning when transactions are accepted by the builder and how many become outdated during the waiting period. This information remains inaccessible to other users. 
However, we assume users share and rate their personal experiences of using different builders on various public forums from which estimates for $F_d$ and $F_m$ may be derived. 
% In this context, community evaluations suggest that builders providing faster service receive better feedback. 
% Scores $F_d$ and $F_m$ can be derived from other users' experiences when shared.

%After getting the rewards ${Score_i}$ from the builders, users can select the builders with the probability of the rewards.
%\begin{align}
%P(\text{select } builder_i) = \frac{Score_i}{sum(Score)}
%\end{align}

% Considering that some users might be reluctant or unwilling to adopt a new policy, we opt to use the score as a weight for probability. 
% Builders with lower scores can still attract partial private transactions.

\subsection{Problem statement}
%% game
%% consider 2 builders, 
%% builder design to use the face by the future validator is known
%% builder can build more block

We aim to design a block building policy $\pi_\mathrm{build}$ and a bidding strategy $\pi_\mathrm{bid}$ for the primary builder such that, the total rewards earned by the primary builder surpasses that of the secondary builder at equilibrium.
Our objective in this work is to only show the existence of an equilibrium where the primary builder wins (vs. the secondary builder). 
We leave the problem of determining the optimal building and bidding policies for maximizing the primary builder's reward for future work. 
% More sophisticated online optimization strategies might yield better results.
% We hope to find an example of strategy to show the space for optimization.

\section{Flashback Builder Design}
\label{s: builderdesign}

We propose a novel block building and bidding strategy for the primary builder, and call our builder design Flashback.\footnote{Henceforth, we use the terms Flashback builder and primary builder interchangeably.}
A Flashback builder selects unconfirmed high-value private transactions available to it at round $t$, and advertises those to the proposer assigned for round $t+1$. 
If the proposer for round $t+1$ desires the advertised transactions, the Flashback builder reserves those transactions for the $(t+1)$-th proposer and includes the reserved transactions in the block built at round $t+1$. 
Transactions reserved this way for future block proposers are sold at a much higher price than normal. 
That is, if $r_2$ is the normal amount of reward earned by a builder per unit transaction fee, a Flashback builder charges $r_1 > r_2$ amount of reward per unit transaction fee for reserved transactions. 
Despite charging a higher-than-normal rate for high-value transactions, future proposers can still be incentivized to reserve transactions from the Flashback builder due to the rarity of high-valued private transactions. 
In the following, we describe the detailed mechanics of a Flashback builder. 
In~\S\ref{s:analysis} we formally analyze our proposed design. 

A Flashback builder at round $t$ considers only the proposer at round $t+1$ for advertising transactions and taking reservations.     
A more general design could consider the Flashback builder at round $t$ interacting with proposers assigned for rounds $t+1, t+2 \ldots, t+c$. 
Such a design can result in an even better reward compared to our present proposal, but arguably is also more complex. 
In fact, even considering only the proposer at round $t+1$, the space of possible block building and transaction auctioning policies is vast. 
Our primary intention in this work is to show the existence of builder designs resulting in an improved equilibrium rewards under the future proposer auction model. 
We hope our work can inspire follow-up research (and implementations) on optimal auction and block construction policies.

\subsection{Private transaction bidding}
\label{s: auction}
\begin{figure}[!t]
\centering
\includegraphics[width=0.75\textwidth]{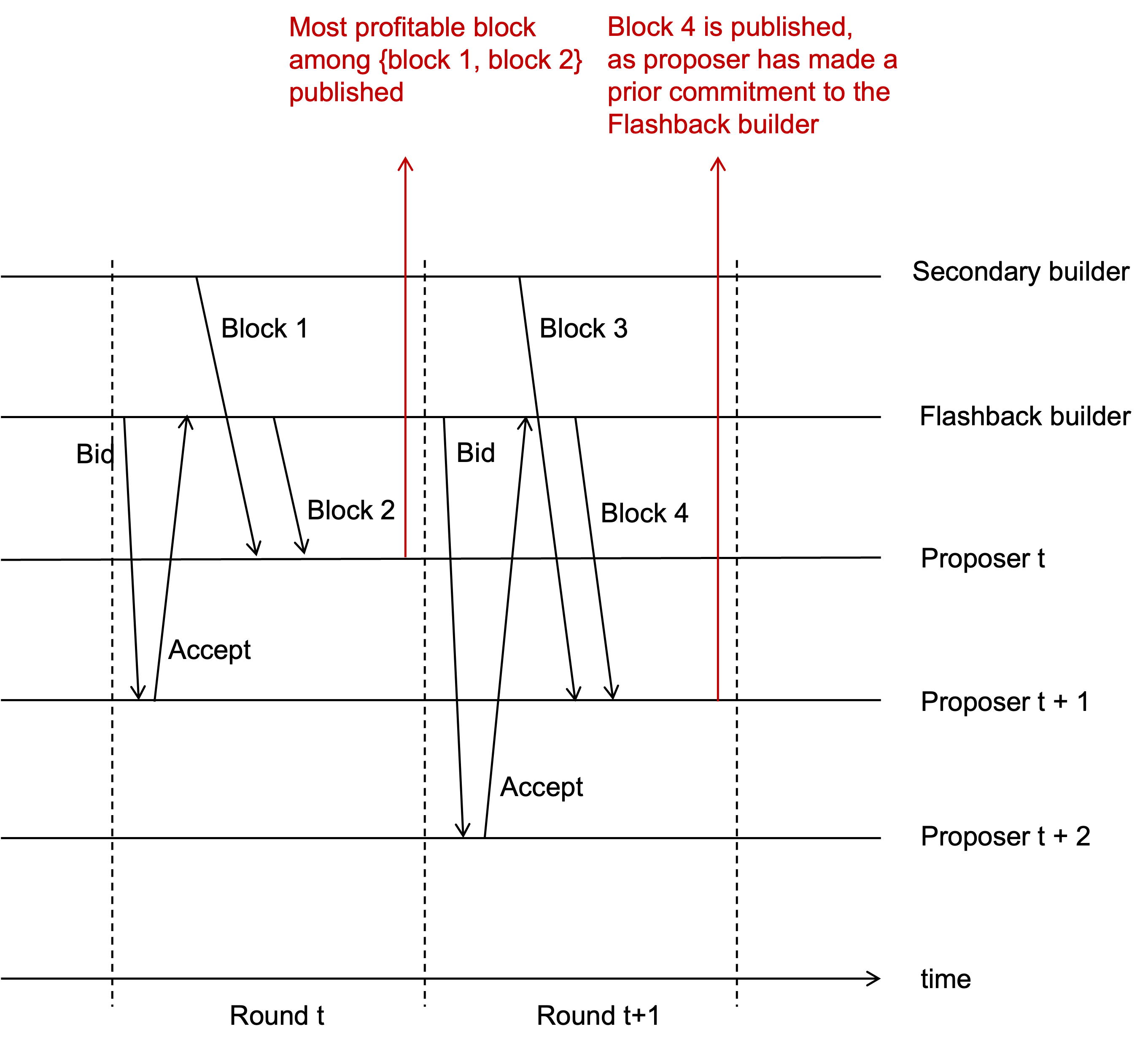}
  \caption{Example showing a timeline of messages exchanged. We assume proposer $t$ has not accepted a bid during time $t-1$. It is also possible for the Flashback builder to not issue a bid during a round.}
  \label{fig:timeline}
\end{figure}

Consider round $t$ in the system. 
Let $Q_t$ denote the set of unconfirmed private transactions available to the Flashback builder at time $t$ such that (1) transactions in $Q_t$ expire only after time $t+1$, and (2) a transaction in $Q_t$ has not been previously auctioned off to proposer $t$.
Let $\rho > 0$ be an estimate of the average reward expected to be earned by a proposer (we discuss how to compute this estimate in \S\ref{s: threshold}).    
For a parameter $k < K$ (recall, $K$ is the block size) let $Q_t^{1:k}$ be the set of $k$ transactions in $Q_t$ that have the highest transaction fees. 
The Flashback builder advertises the transaction fees of the transactions in $Q_t^{1:k}$ and a parameter $0 < r_1 < 1$ to the proposer at time $t+1$ (\S\ref{s: threshold} discusses how to choose $r_1$). 
If the $(t+1)$-th proposer accepts the offered transactions, the Flashback builder reserves the transactions in $Q_t^{1:k}$ to the proposer. 
During the next round, $t+1$, the Flashback builder includes the reserved transactions in the block it builds and sends the block to the proposer of that round. 
While building the block, any remaining space in the block after adding the reserved private transactions is filled up by unconfirmed public transactions.
For transactions that have been reserved ahead of time, the proposer at $t+1$ receives a fraction $(1-r_1)$ of the fees of those transactions; 
while the Flashback builder receives the remaining $r_1$ fraction of the fees. 
For the public transactions, or for private transactions that have not been previously reserved by the $(t+1)$-th proposer, a fraction $r_2$ of the fees is received by the proposer while the builder receives the remaining $(1-r_2)$ fraction. 
% Note that $r_1 > r_2$. 
If the $(t+1)$-th proposer accepts the Flashback builder's bid at time $t$, the $(t+1)$-th proposer commits to publishing the Flashback builder's block at time $t+1$. 
This is the main advantage of offering transactions to the future proposer---if the $(t+1)$-th proposer has committed to receiving a Flashback block at time $t$, then during time $t+1$ even if the secondary builder's block is more profitable the proposer must accept only the Flashback block.
In exchange for reducing risk (of receiving a poor payout in the future) at proposer $t+1$, the Flashback builder gains an upfront commitment from the proposer to accepting a Flashback block.  
If any party deviates from protocol violating the other party's trust, the affected party can choose to stop interacting with the deviant party in the subsequent rounds. 

During time $t$, if the $(t+1)$-th proposer is rational, it accepts the offered transactions from the Flashback builder as long as the profits earned from the transactions exceed (or, significantly exceed) the average profit expected by the proposer. 
Otherwise, the proposer rejects the bid. 
If the bid is rejected, the Flashback builder can still advertise a block to the proposer in the next round $(t+1)$. 
To do this, the Flashback builder resorts to the default block-building policy in which it compiles the $K$ highest value (unreserved, unexpired) private and public transactions known to it and forms a block.  
All transactions in this block are offered at the default rate of $r_2$ (i.e., proposer receives fraction $(1-r_2)$ of reward). 
Figure~\ref{fig:timeline} illustrates a timeline depicting the actions of the primary builder and proposers.

The $(t+1)$-th proposer is likely to accept the bid if the threshold $\rho$ is chosen large enough. 
We discuss how to choose $\rho$ next.

\subsection{Threshold $\rho$ and rate $r_1$}
\label{s: threshold}

For any time $t$, let $V[t]$ be the reward earned by the proposer who published the block at time $t$. 
We estimate the average reward expected to be earned by a proposer at time $t+1$ without accepting a Flashback bid at time $t$ as 
\begin{align}
\hat{E}[V[t+1]] = \frac{\sum_{t' = t - W - 1}^{t-1} V[t'] \mathbf{1}_{\text{Proposer } t' \text{did not accept a Flashback bid at time } t' - 1}}{\sum_{t' = t - W - 1}^{t-1} \mathbf{1}_{\text{Proposer } t' \text{did not accept a Flashback bid at time } t' - 1}},   
\end{align}
where $\mathbf{1}_{\cdot}$ is the indicator function and $W$ is a moving-average window size parameter.
Note the all the information needed to compute $\hat{E}[V[t+1]]$ is available to the Flashback builder from the public blockchain data and logs of past private communication to proposers. 
To encourage proposer $t+1$ to accept the Flashback builder's bid, we choose a threshold $\rho[t]$ as 
\begin{align} 
\rho[t] = (1 + \epsilon) \hat{E}[V[t+1]].
\end{align}
Here, $\epsilon > 1/(1-r_2) - 1$ is a parameter to ensure the rewards gained by proposer $t+1$ if it accepts the bid are significantly higher than the expected reward. 

% \subsection{Auction Rate}
% \label{s: auctionrate}

Let $\mathcal{R}[Q_t^{1:k}]$ be the total amount of fees in the bid transactions $Q_t^{1:k}$. 
From the bidding policy, we must have $\mathcal{R}[Q_t^{1:k}] > \rho[t]$.
The bidding rate $r_1$ is chosen as 
\begin{align}
r_1 = \frac{\mathcal{R}[Q_t^{1:k}] - (1+\epsilon)(1-r_2) \hat{E}[V[t+1]]}{\mathcal{R}[Q_t^{1:k}]}. \label{eq:r1comp}
\end{align}
Since $\mathcal{R}[Q_t^{1:k}] > \rho[t]$, we have $\mathcal{R}[Q_t^{1:k}] > (1+\epsilon)\hat{E}[V[t+1]] > (1-r_2)(1+\epsilon)\hat{E}[V[t+1]]$. 
Equation~\eqref{eq:r1comp} above is therefore well defined.  
Under this choice of $\rho[t]$ and $r_1$, the rewards gained by the $(t+1)$-th proposer $(1-r_1) \mathcal{R}[Q_t^{1:k}]$ from the bid transactions exceed the expected reward $\hat{E}[V[t+1]]$, thus encouraging the proposer to accept the bid.

%\subsection{Builder's Rate}
%IF the primary builder can get a higher reward rate from the next proposer, it can share part of is reward to the users to ease their cost and attract more users to make it as an choice.
%Here we will show that the primary builder can get a higher rewards rate by the strategy, but most of the rate are caught by the users or the proposer.

%Let's show the reward rate distribution as $[r_u, r_b, r_v]$ where $r_u$ is the rate for the users and $r_b$ is the rate for the builder, $r_v$ is the rate for the proposer.
%$r_u + r_b + r_v = 100\%$.
%If $r_b = b$ where $b$ is a significant large rate, which the rate distribution as $[u, b, v]$.
%There can be another builder with the rate distribution as $[u+\epsilon, b-\epsilon, v]$ with $\epsilon > 0$ to share $\epsilon$ to the users so that it can attract more users and can be competitive to the origin builder.
%Also there can be another builder with rate distribution as $[u, b-\epsilon, v+\epsilon]$ with greater rewards to the proposer so the such builder is more welcomed by the proposers and can be selected with higher probability.

%Due to the above analysis, we can find that, if $r_b$ is profitable enough, there can be some other builder with a lower rate to bring more advantage to the users ot builders so as to beat the origin builder. 
%Thus the builder's rate can not be very high.

\section{Analysis}
\label{s:analysis}
In this section we analyze a simplified model of our proposed builder design. 
Let $t \in \{0, 1, 2, \ldots \}$ denotes the round number. 
We assume there are two builders in the system: a primary builder and a secondary builder. 
The primary builder runs our proposed policy and sells transaction bundles to future validators. 
The secondary builder runs the default policy and sells transaction bundles only to current validators. 
Let $X[t]$ denote the total value of transactions arriving privately to the primary builder at round $t$. 
Let $X'[t]$ denote the total value of transactions arriving privately to the secondary builder at $t$. 
Let $Y[t]$ be the total value of public transactions at round $t$. 
These random variables are independent of each other and across time. 
We assume $X[t] \sim \exp(1/\mu_1), X'[t] \sim \exp(1/\mu_2), Y[t] \sim \exp(1/\mu_3)$, where $\mu_1, \mu_2, \mu_3$ denote the expectation of their respective random variables. 
For private transactions, value means how much the transactor pays in direct payment to the builder and block proposer. 
For public transactions, value means the amount of ether that can be extracted by performing an MEV attack on the transaction. 
We assume $\mu_1 + \mu_2 =1$. 

Let $B[t] \in \{ \mathrm{primary, secondary} \}$ denote the builder that builds the block at round $t$. 
$R[t] \in \{0, 1\}$ denotes whether the validator at round $t$ made a prior commitment to receive a block from the primary builder at round $t-1$. 
$R[t] = 1$ if a reservation was made and 0 otherwise. If $R[t] = 1 $ then we must have $B[t] = \mathrm{primary}$, since a validator that has promised to receive the primary builder's block at round $t-1$ must keep up its promise at round $t$.  
$N_b[t]$ is the number of blocks built by builder $b$ at round $t$. 
$N_b[t] = \mathbf{1}_{B[t] = b}$ for $b \in \{ primary, secondary \}$. 

Let $V[t]$ be the total value of block at round $t$ and  let $V_b[t]$ be the value received by builder $b \in \{\mathrm{primary, secondary} \}$ at time $t$.
If a builder $b$ does not build the block at round $t$, we have $V_b[t] = 0$. 
$V_p[t]$ is the value received by validator (proposer) $p$ at time $t$.  

% For a builder $b$ that built a block at time $t$, let $G_b[t] = V_b[t] / V[t]$ be the percentage of the total value in the block gained by the builder. 

\smallskip 
\noindent 
{\bf Builder policy.}
For simplicity, we consider a policy for the primary builder in which {\em all} private transactions received by the builder at round $t$ are offered to the validator at round $t+1$ for reservation. 
Also, assume that private transactions received at round $t$ expire at round $t+1$ unless they have been reserved for round $t+1$. 
Let $0 < r_2 < 1$ be the default fraction of value that builders retain when they build a block. 
For the policy builder, the fraction of value retained by the builder on transactions reserved to a future validator is $r_1$ where $0< r_1 < 1$. 

\smallskip 
\noindent 
{\bf Validator policy.}
If a validator at round $t$ has reserved transactions from the primary builder in the previous round, then the validator must accept the primary builder's block at round $t$. 
If a validator at round $t$ has not committed to accept any transactions during the previous round, then the validator chooses the builder that offers the highest value block at round $t$. 

A key decision a validator for round $t$ makes is whether or not to accept the primary builder's transactions at round $t-1$. 
To do this, a validator uses a parameter $\rho$ which, informally, can be thought of as an expectation on how much reward the validator hopes to earn. 
If the primary builder at time $t$ offers transactions of value $X[t]$ to the validator $t+1$, then the validator accepts the transactions only if $X[t]$ exceeds the threshold $\rho$. 

\smallskip 
\noindent 
{\bf User policy.}
Users choose the builder to submit their private transactions to depending on the performance of the builders. 
The builder that builds more blocks attracts higher-value private transactions compared to the other builder. 
We quantify performance of a builder $b$ by how high $\mathbf{E}[V_b]$, the value gained by the builder $b$, is relative to the other builder. 
% $\mathbf{E}[N_b[t]]$ and the average percentage value gained by the builder are, relative to the respective values achieved by the other builder.  
% A high $\mathbf{E}[N_b[t]]$ means users' transactions to the builder $b$ are likely to succeed and vice-versa. 
A high $\mathbf{E}[V_b]$ means users' transactions to builder $b$ are likely to be confirmed quickly on the blockchain and vice-versa. 
% A high $\mathbf{E}[G_b[t] \mid B[t] = b]$ implies users (searchers) using the builder $b$ stand to profit a greater amount of extracted value and vice-versa. 
% A high $\mathbf{E}[V_b]$ also implies users (searchers) using the builder $b$ stand to profit a greater amount of extracted value and vice-versa. 
% We define performance score of builder $b$ as $S_b = \mathbf{E}[N_b[t]] + \eta \mathbf{E}[G_b[t] \mid B[t] = b]$, where $\eta > 0$ is a parameter. 
We define performance score $S_b$ of builder $b$ as $\mathbf{E}[V_b]$. 
% - expiry time, reward back to user 
The random variables $X[t]$ and $X'[t]$ have distributions that are defined by $S_\mathrm{primary}$ and $S_\mathrm{secondary}$. 
We let 
\begin{align}
\mu_1 &= S_\mathrm{primary} / (S_\mathrm{primary} + S_\mathrm{secondary}) \label{eq:mu1fixedpoint} \\
\mu_2 &= S_\mathrm{secondary}/(S_\mathrm{primary} + S_\mathrm{secondary}). \label{eq:mu2fixedpoint}
\end{align}

\smallskip 
\noindent 
{\bf Game formulation.}
The model described above is a multi-agent game in which the players are the primary builder and all the validators. 
The primary builder's action consists of specifying $r_1$ and $\rho$. 
Note that if $r_1$ is set to be equal to $r_2$ and if $\rho$ is set to be infinity, then the primary builder's policy is exactly the same as the secondary builder's policy.
A validator's action consists of either following our proposed policy or following the default policy. 
If a validator follows our proposed policy, the $t$-th validator commits to accepting  private transactions from the primary builder at time $t-1$ as long as the transactions have value exceeding $\rho$. 
A validator that deviates from our proposed policy does not accept offers ahead of time from the primary builder---the validator simply chooses the best available block at time $t$ and does not make any commitments to the primary builder at time $t-1$. 
We assume $r_2, \mu_3$ are fixed (i.e., part of the environment) and cannot be controlled by the players. 
The objective for the primary builder is to achieve a score that is greater than the score of the secondary builder. 
The objective of a validator is to receive the highest value blocks from the builders. 

\subsection{Validator and Builder Rewards}

Under the builder, user and validator policies mentioned previously, let $V_p^\mathrm{policy}[t]$ denote the value earned by the validator at time $t$. 
We have 
\begin{align}
V_p^\mathrm{policy}[t] &= \mathbf{1}_{R[t] = 1}((1-r_1) X[t-1] + (1 - R[t+1]) (1-r_2) X[t] + (1-r_2) Y[t]) \notag \\ 
& ~~~~~ + \mathbf{1}_{R[t]=0}\max \{  (1-r_2) X'[t] + (1-r_2) Y[t], \notag \\ 
& ~~~~~ (1-R[t+1])(1 - r_2)X[t] + (1-r_2)Y[t] \} 
\end{align}
The expected value earned by a validator is given in Appendix~\ref{apx:vppolicy}

Next, let $V_p^\mathrm{default[t]}$ be the value earned by the validator at time $t$ that is not following our proposed policy, but instead follows the default policy. 
However, we assume the validators at all other time steps---specifically at time $t-1$ and $t+1$,  follow our proposed policy.
Analyzing $V_p^\mathrm{default}$ tells us whether validators have an incentive to deviate from our proposed protocol. 
We have 
\begin{align}
V_p^\mathrm{default}[t] = \max \{  (1-r_2) X'[t] + (1-r_2) Y[t], (1-R[t+1])(1 - r_2)X[t] + (1-r_2)Y[t] \}. 
\end{align}
As before, we compute the expected value in Appendix~\ref{apx:vpdefault}.

Similarly, we next compute the builder rewards. 
Let $V_\mathrm{primary}[t]$ be the value earned by the primary builder at time $t$. 
$V_\mathrm{secondary}[t]$ is the value earned by the secondary builder at $t$. 
We have 
\begin{align}
V_\mathrm{primary}[t] &= \mathbf{1}_{R[t] = 1}(r_1 X[t-1] + (1 - R[t+1]) r_2 X[t] + r_2 Y[t]) \notag \\
& ~~~~~ + \mathbf{1}_{R[t]=0} \mathbf{1}_{(1-r_2) X'[t] + (1-r_2) Y[t] < (1-R[t+1])(1 - r_2)X[t] + (1-r_2)Y[t]} \notag \\ 
& ~~~~~~~~ ((1-R[t+1]) r_2 X[t] + r_2 Y[t]).  
\end{align}

The expected rewards earned by the primary builder is given in Appendix~\ref{apx:vprimary}.

The value earned by the secondary builder is given by
\begin{align}
V_\mathrm{secondary}[t] &= \mathbf{1}_{R[t]=0} \mathbf{1}_{(1-r_2) X'[t] + (1-r_2) Y[t] > (1-R[t+1])(1 - r_2)X[t] + (1-r_2)Y[t]} r_2(X'[t] + Y[t]),  
\end{align}
whose expectation is as in the proposition in Appendix~\ref{apx:vsecondary}. 

\subsection{Equilibrium Analysis}
\label{s:Equilibrium Analysis}

In the following we show that there exists a Nash equilibrium with $\rho < \infty$ for the primary builder, and where validators follow our proposed policy. 
From Equations~\eqref{eq:mu1fixedpoint} and~\eqref{eq:mu2fixedpoint} for a fixed $r_1$ and $\rho$, the value of $\mu_1$ and $\mu_2$ is given by the following fixed point equations. 
\begin{align}
\mu_1 = \frac{\mathbf{E}[V_\mathrm{primary}[t]]}{\mathbf{E}[V_\mathrm{primary}[t]] + \mathbf{E}[V_\mathrm{secondary}[t]]}, 
\mu_2 = \frac{\mathbf{E}[V_\mathrm{secondary}[t]]}{\mathbf{E}[V_\mathrm{primary}[t]] + \mathbf{E}[V_\mathrm{secondary}[t]]}, \label{eq:fixedpointmu1mu2e}
\end{align}
where $\mathbf{E}[V_\mathrm{primary}[t]]$ and $\mathbf{E}[V_\mathrm{secondary}[t]]$ are as in Propositions~\ref{prop:primarybuilder} and~\ref{prop:secondarybuilder} respectively. 

For the primary builder to achieve a score greater than the secondary builder, we must have 
\begin{align}
\mathbf{E}[V_\mathrm{primary}[t]] - \mathbf{E}[V_\mathrm{secondary}[t]] &> 0 
\iff \frac{\mu_1}{\mu_2}\mathbf{E}[V_\mathrm{secondary}[t]] - \mathbf{E}[V_\mathrm{secondary}[t]] > 0 \notag \\
\iff \mu_1 &> \mu_2,  \label{eq:mu1bigmu2rew}
\end{align}
where the second inequality above follows from Equation~\eqref{eq:fixedpointmu1mu2e}. 
Since $\mu_1 + \mu_2 = 1$, a solution (i.e., a $r_2, \rho$ value) to the fixed point Equation~\eqref{eq:fixedpointmu1mu2e} where $\mu_2 < 1/2$ guarantees the primary builder to achieve a score that is greater than that of the secondary builder. 

For a validator to follow our proposed policy and not deviate back to the default policy, we must have $\mathbf{E}[V_p^\mathrm{policy}[t]] > \mathbf{E}[V_p^\mathrm{default}[t]]$. 
The following proposition shows a sufficient condition for this. 
\begin{lemma}
\label{lem:validatorpolicymucondi}
For $r_1 = 0, \rho > \mu_2$ and $\mu_2 < 1/2$, we have $\mathbf{E}[V_p^\mathrm{policy}[t]] > \mathbf{E}[V_p^\mathrm{default}[t]]$. 
\end{lemma}
(Proof in Appendix~\ref{apx:validatorpolicymucondi})

From Equation~\eqref{eq:mu1bigmu2rew}, it therefore suffices if we can show there exists a solution to the fixed point equations in Equation~\eqref{eq:fixedpointmu1mu2e} with $\mu_2 < 1/2$ and $\rho >\mu_2$.  
We rewrite Equation~\eqref{eq:fixedpointmu1mu2e} as 
\begin{align}
\mathbf{E}[V_\mathrm{primary}[t]]\mu_2 - \mathbf{E}[V_\mathrm{secondary}[t]] \mu_1 =0, \label{eq:emuminuemu}
\end{align}
with $\mu_1 = 1 - \mu_2$. 
To show the existence of a fixed point with $\mu_2 < 1/2$, we first show that at $\mu_2 = 1/2$ and for sufficiently large $\rho$, the left hand side of Equation~\eqref{eq:emuminuemu} is negative. 
\begin{lemma}
\label{lem:fixedptlesszero}
For $\mu_2 = 1/2, r_1=0$ and $\rho > \ln(3)/2$ we have $\mathbf{E}[V_\mathrm{primary}[t]]\mu_2 - \mathbf{E}[V_\mathrm{secondary}[t]] \mu_1 < 0$.
\end{lemma}
(Proof in Appendix~\ref{apx:fixedptlesszero})

Next, we show there exists a $\mu_2 < 1/2$ where $\mathbf{E}[V_\mathrm{primary}[t]]\mu_2 - \mathbf{E}[V_\mathrm{secondary}[t]] \mu_1$ is positive for sufficiently large $\rho$.
\begin{theorem} 
\label{thm:fixedpointlesshalf}
For $r_1 = 0$ and sufficiently large $\rho$, Equation~\eqref{eq:emuminuemu} has a fixed point solution with $0<\mu_2<0.5$.  
\end{theorem}
(Proof in Appendix~\ref{apx:fixedpointlesshalf})

Combining Lemma~\ref{lem:validatorpolicymucondi} and Theorem~\ref{thm:fixedpointlesshalf} we conclude there exists a $\rho < \infty$ and $r_1 = 0$ where the primary builder and validators are at equilibrium and do not have an incentive to deviate from protocol. 
An $r_1$ value of $0$ means the primary builder does not earn any reward from private transactions bid to future proposers. 
Despite this, the analysis above shows that the overall rewards earned by the primary builder exceed that of the secondary builder. 
By bidding high-value transactions at a discounted rate to (future) proposers and obtaining upfront commitments, we increase the chance of proposers publishing blocks built by the primary builder. 
This in turn has the effect of increasing the score of the primary builder, and consequently attracts more high-value private attractions to be sent to the primary builder.
While some of the private transactions are reserved for future proposers, the ones that are not reserved earn a reward at a rate of $r_2$ for the primary builder resulting in a net-positive effect for the builder. 
The analysis highlights the complex interplay between various factors affecting total rewards earned, and shows how policies can be  counter-intuitive yet beneficial. 

\section{Evaluation}
\label{s: evaluation}

\subsection{Experiment Setup}
\label{s: setup}
To align our simulation on Flashback with real-world data distribution, we compiled a dataset comprising 10,000 blocks spanning from block number 15,570,981 to 15,580,985. These blocks collectively contain 157,946 transactions. Our dataset consists of a 7-part set for each transaction, including details such as transaction hash, sender and receiver addresses, direct payment fee, transaction fee, gas price, and gas used.
Additionally, for each block, we generated a 3-part dataset comprising the block number, the list of transactions within the block, and the base fee. These data points were meticulously extracted from a credible source—specifically, Etherscan~\cite{etherscan}. Transactions within each block represent successfully mined transactions, accompanied by their unique hash and index within the block.
To distinctly identify private and public transactions, we carefully considered 2,740 instances of private transactions, constituting 1.73\% of the overall transaction count. This differentiation was achieved through cross-referencing another reputable source, Zeromev~\cite{zeromev}.
\begin{figure}[htp]
  \centering
    \includegraphics[width=0.8\linewidth]{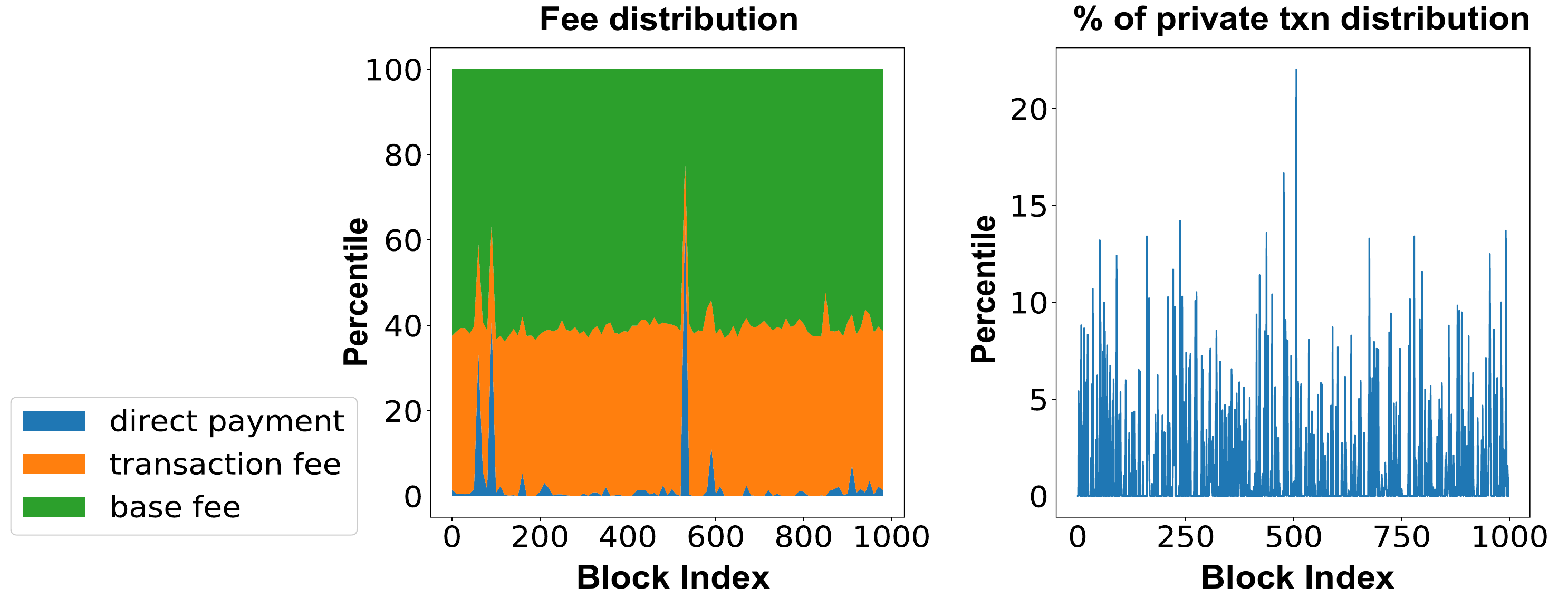}
    \caption{Share of user payments and percentile of private transactions} \label{fig: fee_distribution}
\end{figure}

Figure~\ref{fig: fee_distribution} illustrates the fee distributions among blocks in the dataset, which fits well to an exponential distribution with around 3\% being private transactions among all the transaction and around 10\% of the profits come from direct payments, constituting a proportion of all profits derived from both direct payments and transaction fees.
Although private transactions represent only 3\% of the total transaction volume, they contribute to over 10\% of the overall profits. 
This observation suggests that private transactions tend to yield higher profitability compared to public transactions.

Flashback provides the primary builder with two actions. 
The first action involves sending bid messages to the proposer in the subsequent round if the primary builder identifies that the current private transactions meet a high standard of quality, denoted as $\rho$.
Initially, we conduct experiments without placing bids and record the distribution of private transaction fees. 
Subsequently, based on the bidding strategy for private transactions, we estimate the fees using recorded data. 
During this process, we document the distribution of private transaction values for proposers who do not accept bids.
The bid rate $r_1$ is determined by the threshold $\rho$ in such a way that for bids valued at $\rho$, the proposer can still attain profits higher than their expectations.

At the user's end, we employ a scoring function model to evaluate builders based on historical performance. 
Builders with higher scores are preferred by users, increasing the likelihood to be chosen.
However, we acknowledge that some users might not be as sensitive to scores and could select poorly performing builders despite their scores.
In the initial phase of the experiment, users possess an initial score for the builders, which influences their selection in the first round. 
We've set the initial knowledge length to 1 for easy overwriting.
In Section~\ref{s: bid strategy}, we extend the duration of initial knowledge, setting to 200 rounds.

Private transactions hold time sensitivity, as users anticipate prompt processing. Hence, we introduce a Time-to-Live (TTL) parameter for all private transactions. These transactions are set to expire if they remain in the private pool for TTL rounds.
Builders can retain private transactions for a maximum of TTL rounds, providing them an incentive to process these transactions promptly to extract transaction fees. 
The average waiting time plays a pivotal role in how users assess builders. Users tend to favor builders capable of processing their transactions promptly or, at the very least, preventing their private transactions from expiring.
Initially, we've set TTL to 10 rounds. Later in Section~\ref{s: bid ttl}, we will explore varying TTL to analyze its impact on the game.

The processing delay and missing rate of transactions remain private, as these transactions are sent through private channels. 
We make an assumption that users can share their experiences with these builders and assess their performance based on community reviews. 
These experiences are publicly accessible within the community, allowing users to factor them into their scoring process.
For example, in a block containing 100 transactions, if 10\% of users are willing to provide feedback, they can broadcast their waiting times or whether their transactions became outdated to the community. 

\subsection{Experiment Result}
\label{s: result}

\begin{figure}
  \centering
  \begin{subfigure}[b]{0.75\textwidth}
    \includegraphics[width=\linewidth]{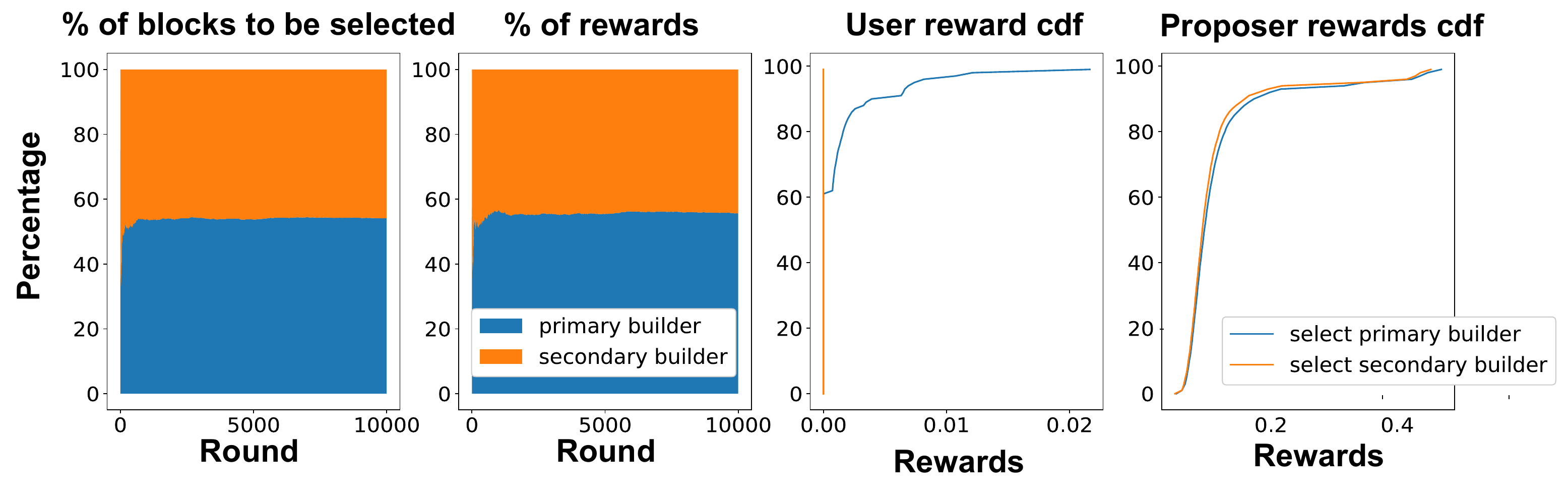}
    \caption{Rewards distributions for builder, user and proposer} \label{fig:top3_ttl10}
  \end{subfigure}%
  \hspace*{\fill}
  \begin{subfigure}[b]{0.23\textwidth}
    \includegraphics[width=\linewidth]{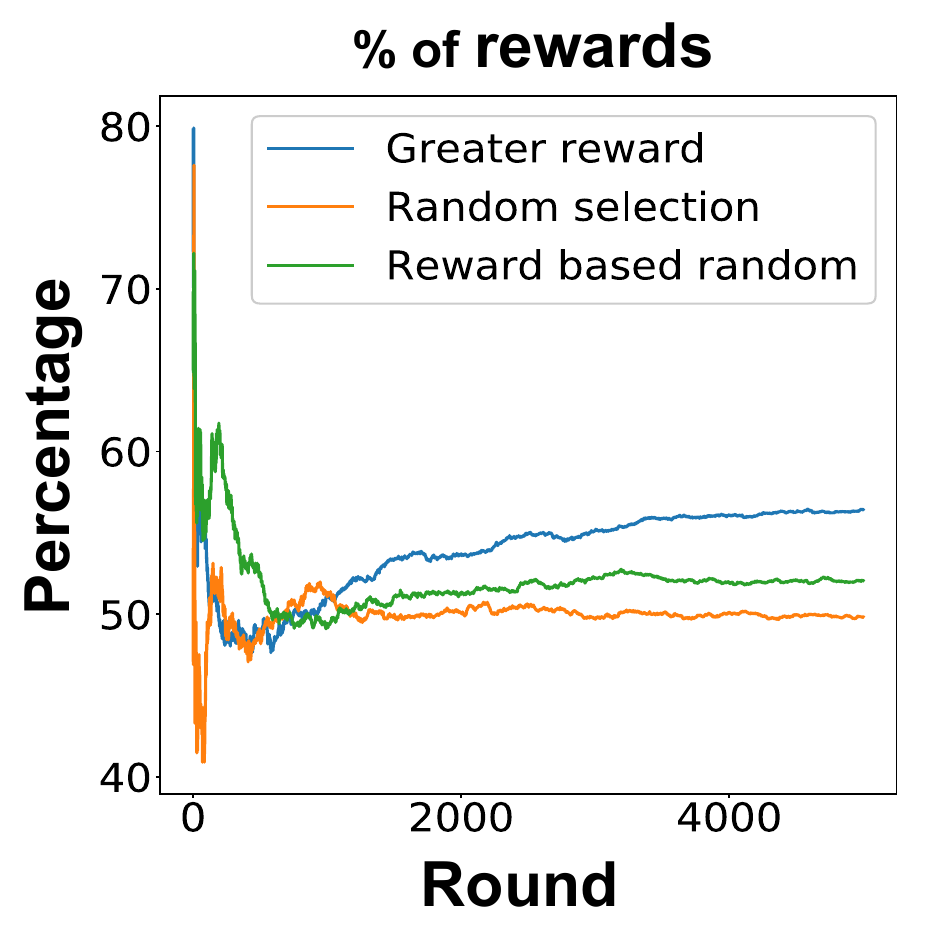}
    \caption{Primary builder's rewards distribution} \label{fig: algobid}
  \end{subfigure}%
  \caption{Rewards distributions}
\label{fig:RewardsDistributions}
\end{figure}

We first conduct the simulations as in Figure~\ref{fig:top3_ttl10}, where the primary builder is allowed to bid their top 3 most valuable private transactions to the proposer in the subsequent round. Additionally, all transactions have a TTL of 10 rounds, after which they expire.
The builder's design and policy strategy are outlined in Section~\ref{s: builderdesign}.
We execute the experiments over 10,000 rounds, observing the network state convergence at approximately 1000 rounds.

In Figure~\ref{fig:top3_ttl10}, we present the builder's block ratio and reward distributions for builder, user and proposer.
In the first subfigure, we illustrate the block ratio attributed to the primary builder and the secondary builder. After 500 rounds, the primary builder gains an advantage, with approximately 55\% of the blocks attributed to them, while the secondary builder holds 45\% of the blocks.
The second subfigure depicts the builder's reward distribution according to the score function. 
As builders receive fixed-rate rewards for packaging blocks, their rewards are closely correlated with the percentile of blocks they construct. 
Notably, the primary builder's rewards surpass those of the secondary builder by 20\%, , calculated as 55\% compared to 45\%.
The third subplot plots the CDF of users' rewards.
Primary builder gets $r_1$ through $\pi_\mathrm{build}$.
Around 40\% of the users who send their private transactions to the primary builder get part of the cost back.
Conversely, the secondary builder adheres to the default strategy where users cover all transaction fees and direct payments, resulting in a constant curve at 0 (orange curve).
The forth subplot plots the CDF of proposers' rewards, proposer who select primary builder's block gets a greater rewards distribution.
In conclusion, all the primary builders, user and proposers all experience enhanced rewards through $\pi_\mathrm{build}$.
The state of convergence after 1000 rounds demonstrates a more favorable equilibrium among these three players.

Figure~\ref{fig: algobid} illustrates the reward distributions for the primary builder, influenced by 3 simple bidding acceptance policies adopted by the target proposers.
When the proposer prioritizes offers with greater rewards, which is applied in Flashback, the primary builder stands to receive increased rewards.
In cases where the proposer accepts bids randomly, such as with a probability of 50\%, the proposer's rewards tend to approach the 50\% mark.
When employing a bidding acceptance policy that involves randomness based on reward comparisons, the primary builder's rewards fall between the two aforementioned scenarios.

\subsubsection{Initial state}
\label{s: initial start}
The aforementioned experiment initializes equal scores for both builders, where users initially lack preferences and treat the builders impartially.
Currently, Flashbot contributes approximately 70\% of blocks on the Ethereum network, indicating significant user attraction to their services.
We also explore a scenario of challenging initiation, where we introduce a primary builder to the existing network. 
Initially, users exhibit a strong preference for the established secondary builder, resulting in the primary builder having a substantially lower initial score in comparison.
We proceed with experiments involving primary builder scores set at 1, 1/2, 1/4, 1/8, and 1/16 times the initial score of the secondary builder. 
The length of the score period spans 200 rounds, ensuring the initial score's lasting influence over an extended duration.

\begin{figure}[!tbp]
  \centering
  \begin{subfigure}[b]{0.48\linewidth}
    \includegraphics[width=\linewidth]{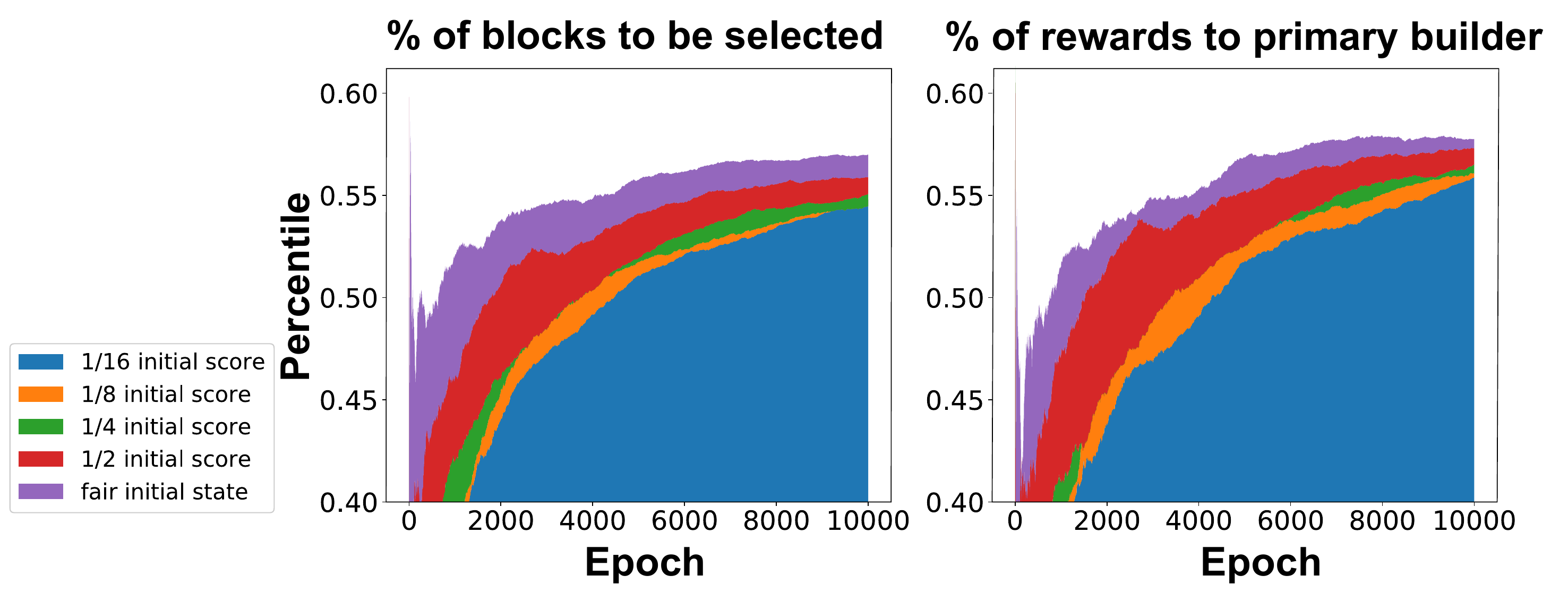}
    \caption{Primary builder starts with lower score} 
    \label{fig: poorstart}
  \end{subfigure}%
  \begin{subfigure}[b]{0.5\linewidth}
    \includegraphics[width=\linewidth]{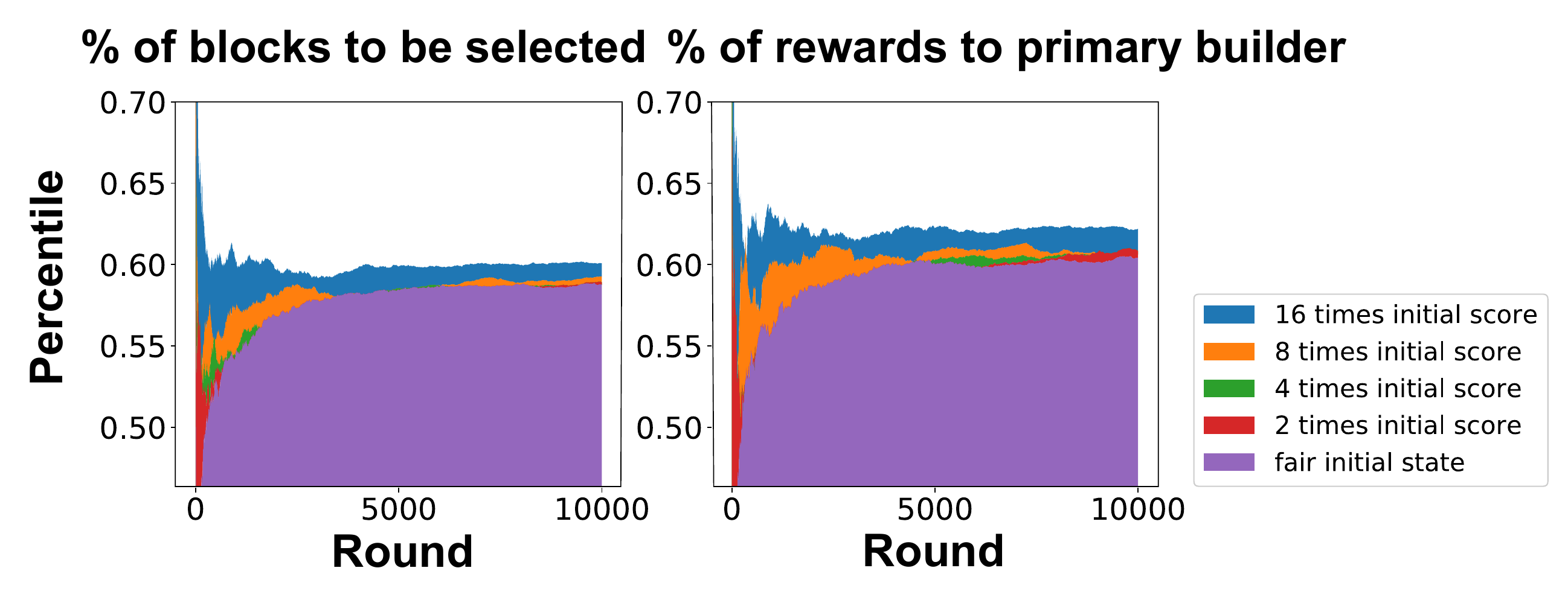}
    \caption{Primary builder starts with higher score} 
    \label{fig: richstart}
  \end{subfigure}%
  \caption{Primary builder starts with different initial states}
\label{fig:initial state}
\end{figure}

In Figure~\ref{fig: poorstart}, We observed that primary builders with lower initial preference states take a longer time to reach the convergence state. Notably, the blue line, representing a primary builder with only 1/16 of the score, has not fully converged even after 10,000 rounds.

In contrast to the aforementioned scenario, we explore another case where a new primary builder experiences a surge in popularity, with users exhibiting significantly greater preference for them.
In these experiments, we initiate scenarios with primary builder scores set at 1, 2, 4, 8, and 16 times the initial score of the secondary builder. The results, depicting the percentile of blocks or rewards for both builders, are illustrated in Figure~\ref{fig: richstart}.

Higher initial scores yield higher percentiles, especially in the initial 4000 rounds. However, they eventually converge to around 60\%, as the fix point proofed in full paper~\cite{FlashbackFullPaper}.
While the initial state can influence the early rounds, it's evident that the settings of  $\rho$ and $r_1$ ultimately guide the system towards a comparable convergence state, highlighting the robustness of the observed dynamics.

Setting a substantial difference in initial scores, leading all private transactions to a single builder, prevents the other from updating its performance. 
These cases might represent additional fixed points, but their realization in reality is challenging.

\subsubsection{Different bid strategy}
\label{s: bid strategy}
In previous sections, we examine a scenario where all transactions expire after 10 rounds, and the primary builder exclusively bids on the top 3 private transactions.
In this section, we delve into various bid strategies, focusing on the number of transactions to bid and the timing of private transaction expiration. These bid strategies have the potential to influence users' score functions, thereby inducing significant changes in each player's rewards.
We begin by investigating bid strategies that involve bidding on 1, 2, 3, up to 20 private transactions, with the bid threshold being linked to the length of the bid transactions.
All other network settings remain consistent with those outlined in Section~\ref{s: result}.

\begin{figure}[htp]
  \centering
    \includegraphics[width=\linewidth]{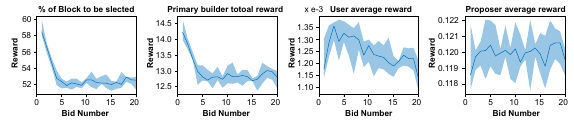}
  \caption{Different policy strategy to bid 1 to 20 transactions}
  \label{fig:bid_num}
\end{figure}

Figure~\ref{fig:bid_num} presents plots depicting the percentile of blocks built by the primary builder, the total rewards for the primary builder, the average rewards for users, and the average rewards for proposers.
A lower bid number simplifies the bidding process, resulting in peaks at $bid_{number} = 1$ in the first and second subfigures.
The third and fourth subfigures demonstrate that a bid strategy involving around 3 transactions yields the highest rewards for users and proposers, with trade-of between the probability to bid and profits per successful bid.

\subsubsection{Different valid time}
\label{s: bid ttl}
In the previous experiments, we assume the TTL for private transaction is 10 rounds that private transactions can be kept for at most 10 rounds to be processed.
In this section, we vary the TTL to investigate how the policy performs under different levels of time sensitivity, while keeping other settings the same as outlined in Section~\ref{s: result}.

\begin{figure}[htp]
  \centering
    \includegraphics[width=\linewidth]{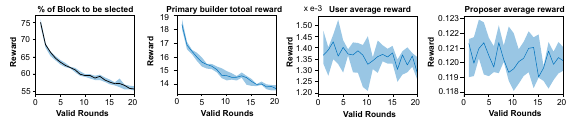}
  \caption{Different network setting for private transactions to keep valid}
  \label{fig:various_ttl}
\end{figure}

Figure~\ref{fig:various_ttl} illustrates that as the TTL increases, the likelihood of the primary builder's blocks being selected diminishes. Additionally, the average rewards for primary builders, users, and proposers collaborating with the primary builder experience a decline.
For $TTL = 1$, which is reflects analysis section, primary builder has the greatest advantage that it can used the best of the private transaction's expire time to insure them to be process in the next round, where users can have more confidence in the primary builder, and primary builder can be competitive compared with the secondary builder.
As TTL increases, the primary builder's advantage from pre-communicating with the proposer in the next round diminishes. 
However, there remains an advantage in the bidding process, with the advantages converging when TTL exceeds 10.

\subsubsection{Policy with 0 profit in bid}
\label{s: zero bid}
We discovered a viable range for the threshold $\rho$—greater than the default rate $r_2$ but less than an upper bound—enabling equilibrium among the three players.
Further exploration revealed an intriguing possibility: a bid rate lower than the builder's default rate can still yield greater rewards for the builder. 
In this scenario, we set the bid rate to 0 while retaining the default rate at the value detailed in Section~\ref{s: result} (2\%).
However, as a consequence of the primary builder receiving lower rewards compared to previous cases. 
This lack of distinction for users in their builder selection leads to equal probabilities for users to choose builders, maintaining consistency with other settings outlined in Section~\ref{s: result}.

\begin{figure}[!tbp]
  \centering
  \begin{subfigure}[b]{0.56\linewidth}
    \includegraphics[width=\linewidth]{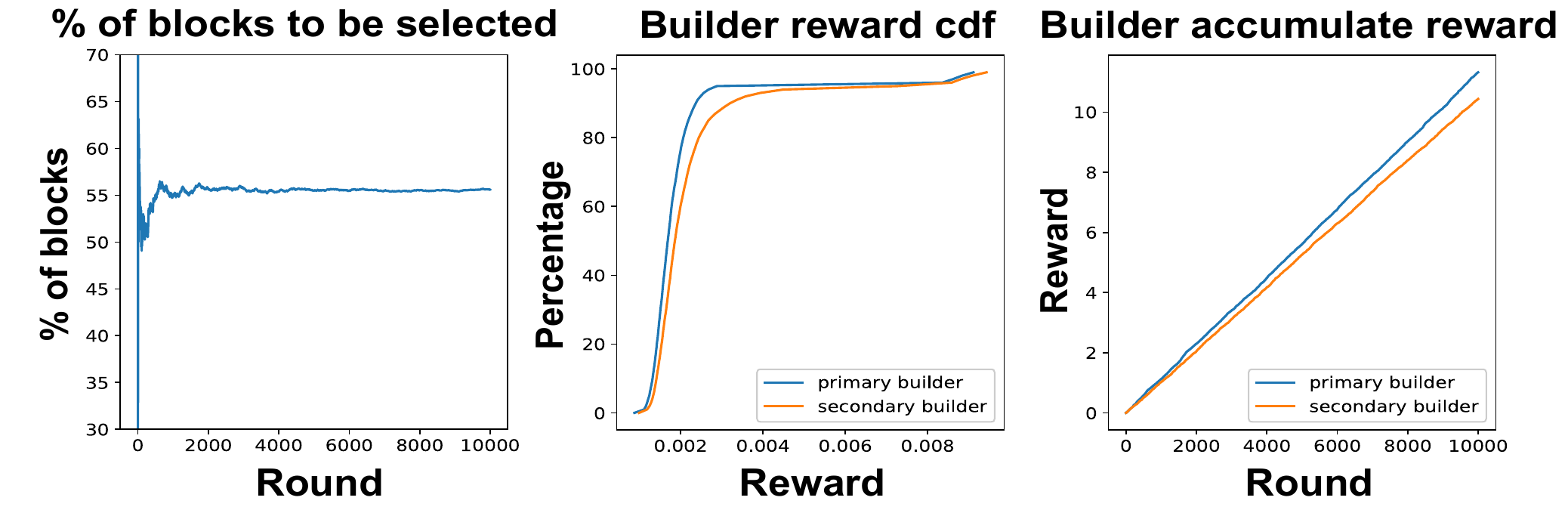}
    \caption{Special policy to set 0 rate in bidding}
    \label{fig:zero_bid}
  \end{subfigure}%
  \begin{subfigure}[b]{0.38\linewidth}
    \includegraphics[width=\linewidth]{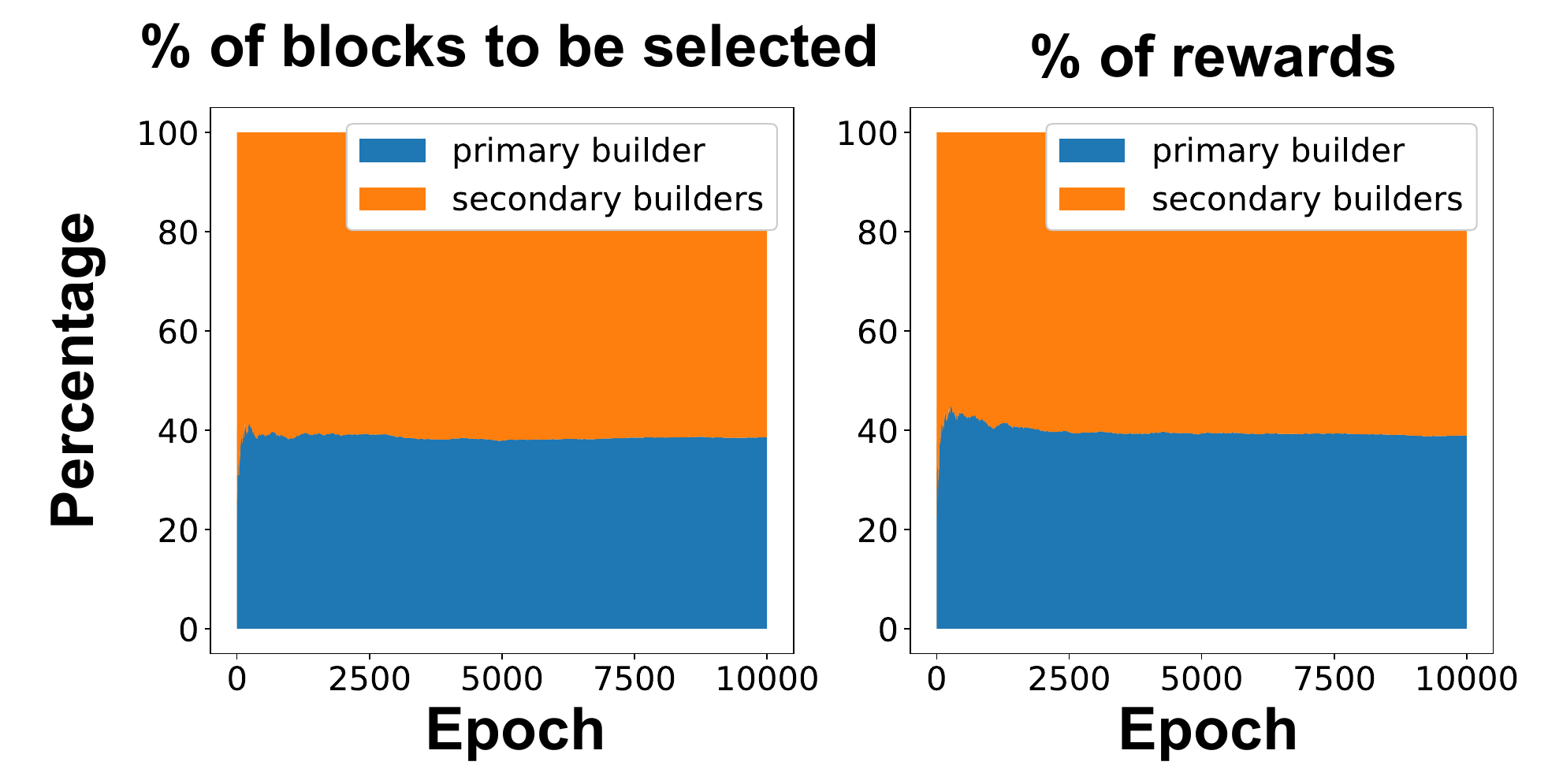}
    \caption{One primary builder and two secondary builder}
    \label{fig:3builders}
  \end{subfigure}%
  \caption{Special network settings}
\label{fig:special cases}
\end{figure}

With an optimal threshold in place, proposers show a distinct preference for the primary builder's candidate block, with greater probability to select primary builder's block as shown in the first plot.
Due to the 0 rate in the bids, primary builder has a lower rewards cdf compared with the secondary builder in the second plot.
Finally, as we combine the probability to be selected and the profits in each selection, primary builder could still achieve a greater accumulate rearwards as plot in the third plot.

\subsubsection{Multiple Secondary Builders}
\label{s:secondarys}
Previously, our discussions focused on the dynamic between a primary builder and a single secondary builder. 
In this subsection, we extend our analysis to encompass scenarios where a primary builder interacts with multiple secondary builders. 
The network settings remain consistent with those described earlier, with the introduction of additional secondary builders participating in the competition for candidate blocks. 
As depicted in Figure~\ref{fig:3builders}, the primary builder can secure more than one-third of the selected blocks and associated rewards, demonstrating a clear advantage over the secondary builders.

\section{Related work}
\subsection{Game theory in blockchain}
In decentralized blockchain networks, users have the opportunity to engage in various roles within the transaction processing and block generation processes. 
These diverse roles can be likened to players in a game, with participants applying game theory to maximize their rewards. 

In the context of POW networks, extensive game theory research has been conducted on block mining.
This research encompasses games related to transaction queues~\cite{li2018transaction}, miner mining~\cite{kiayias2016blockchain}, and mining pool selections~\cite{liu2018evolutionary}. 
The insights gained from these game studies provide players with strategies from learning the networks~\cite{dong2023graph} to optimize their rewards. Strategy design can be applied to various facets, including neighbor selection~\cite{mao2020perigee, xue2023goldfish}, neighbor degree management~\cite{math11234741}, direct miner connection~\cite{vedula2023cobalt} and data storage~\cite{zhang2023Kadabra}.

With Ethereum’s transition from POW to POS, there has been some recent work measuring the adoption and behavior of proposer-builder separation.
Some work finds that proposer may not receive the optimized value as expected~\cite{heimbach2023ethereum}, which fits Flashback's equilibrium.
Another work also points that the conflicts may benefit particular parties due to the implicit trust assumptions~\cite{wahrstatter2023time}.

Proof-of-stake introduces a novel consensus mechanism to the network, characterized by distinct roles in generating new blocks.
Pos network does not need miners any more, but the stake based validator decision also brings a game~\cite{sahin2022optimal, huang2021rich} on how to select the validators based on stakes.
~\cite{bhudia2023game} studies the extortion attacks with game among attacker, victim and validators.
Price of MEV~\cite{mazorra2022price} formalizes a game on transaction ordering mechanism based on priority gas auction and measures its Nash equilibrium, but they don't consider future block auctions.

\subsection{MEV auction platform/flashbots related work}
This paper studies MEV auction platform by applying a greedy strategy in selecting transactions on the validator side.
Flashback's private channel design is based on Flashbot~\cite{flashbots} and Flashback's validator's selection on candidate block is based on MEV-boost~\cite{mev-boost} with connections to multiple relays to search for maximum MEV.
But compared with them, we have simplifications on the block space ordering and functions of relay.

Multi-block MEV~\cite{jensen2023multi} secures MEV in k-consecutive blocks.
Aequitas~\cite{kelkar2020order}, Themis~\cite{kelkar2021themis} arise an ordering strategy with consideration of the received timestamps, which focus more on fairness in transaction ordering.

\subsection{Private transaction}
Prior research~\cite{qin2021quantifying, capponi2022evolution, piet2022extracting, weintraub2022flash, lyu2022empirical} has delved into the realm of private transactions, centering on the assessment of Miner Extractable Value (MEV) and blockchain extractable value (BEV) within the context of private transactions. Notably, \textit{Lyu et al.}~\cite{lyu2022empirical} compiled a year-long dataset of private transactions within PoW Ethereum, conducting an empirical analysis of their characteristics, economic implications (e.g., transaction costs and miner earnings), and security effects. In contrast to these efforts, our study endeavors to introduce a game-theoretic model aimed at redistributing profits among various parties (e.g., builders and validators) in PoS Ethereum. 

There are some work researching to order the transactions by the other rules.
For example, order-fairness~\cite{kelkar2020order} introduces a method to process the transaction based on their received timestamp, Wendy~\cite{kursawe2020wendy} presents a method to order based on the transaction's observation time by honest nodes, 
However, validators directly get rewards from the fees, there are less incentives for validators to apply the other rules.

\section{Conclusion}
This paper addresses the profit distribution in a blockchain ecosystem as a game involving users, builders, and proposers.
We introduce 'Flashback', a novel design aimed at enabling builders to communicate with upcoming proposers when they possess high-quality private transactions. 
The paper conducts a theoretical analysis of this game, establishing equilibrium conditions between primary builders and proposers for specific threshold values $\rho$ and auction rates $r_1$. 
The analysis lays out the conditions necessary for equilibrium between primary builders and proposers and demonstrates through experiments the existence of such equilibrium. 
The findings emphasize the advantages enjoyed by players who adopt the 'Flashback' policy, showcasing improved rewards compared to the default strategies currently in use within the blockchain ecosystem.

%%
%% The next two lines define the bibliography style to be used, and
%% the bibliography file.

%\bibliographystyle{ACM-Reference-Format}
\bibliography{paper}

\begin{thebibliography}{10}

\bibitem{posdoc}
Proof-of-stake (pos).
\newblock
  \url{https://ethereum.org/en/developers/docs/consensus-mechanisms/pos/}.

\bibitem{flashbots}
Flashbots auction.
\newblock \url{https: //docs.flashbots.net/flashbots- auction/overview}, 2022.

\bibitem{mev-boost}
Mev-boost github.
\newblock \url{ https://github.com/ flashbots/mev- boost}, 2022.

\bibitem{builddomandseardep}
Builder dominance and searcher dependence.
\newblock \url{
  https://frontier.tech/builder-dominance-and-searcher-dependence}, 2023.

\bibitem{defiusecases}
Decentralized finance use cases.
\newblock \url{
  https://www.finextra.com/blogposting/21958/top-10-decentralized-finance-defi-use-cases},
  2023.

\bibitem{defimarketcommentary}
Defi market commentary.
\newblock \url{
  https://consensys.net/blog/cryptoeconomic-research/defi-market-commentary-january-2023/},
  2023.

\bibitem{FlashbackFullPaper}
Flashback: Enhancing proposer-builder design with future-block auctions in
  proof-of-stake ethereum.
\newblock
  \url{https://anonymous.4open.science/r/Flashback-3F7C/Flashback\%20full\%20paper.pdf},
  2023.

\bibitem{mevgalaxy}
Maximal extractable value.
\newblock \url{
  https://www.galaxy.com/insights/research/mev-the-rise-of-the-builders/},
  2023.

\bibitem{mevlandscape}
Mev landscape.
\newblock \url{
  https://mirror.xyz/ratedw3b.eth/xltDyY32_xQyDiHHphj230kc0iEQL43Pwb7b6t_fzWQ},
  2023.

\bibitem{eth2book}
A technical handbook on ethereum's move to proof of stake and beyond.
\newblock \url{ https://eth2book.info/capella/part2/building_blocks/}, 2023.

\bibitem{bhudia2023game}
Alpesh Bhudia, Anna Cartwright, Edward Cartwright, Darren Hurley-Smith, and
  Julio Hernandez-Castro.
\newblock Game theoretic modelling of a ransom and extortion attack on ethereum
  validators.
\newblock {\em arXiv preprint arXiv:2308.00590}, 2023.

\bibitem{Buterin2013}
Vitalik Buterin.
\newblock Ethereum white paper: A next generation smart contract \&
  decentralized application platform.
\newblock 2013.
\newblock URL: \url{https://github.com/ethereum/wiki/wiki/White-Paper}.

\bibitem{capponi2022evolution}
Agostino Capponi, Ruizhe Jia, and Ye~Wang.
\newblock The evolution of blockchain: from lit to dark.
\newblock {\em arXiv preprint arXiv:2202.05779}, 2022.

\bibitem{daian2020flash}
Philip Daian, Steven Goldfeder, Tyler Kell, Yunqi Li, Xueyuan Zhao, Iddo
  Bentov, Lorenz Breidenbach, and Ari Juels.
\newblock Flash boys 2.0: Frontrunning in decentralized exchanges, miner
  extractable value, and consensus instability.
\newblock In {\em 2020 IEEE Symposium on Security and Privacy (SP)}, pages
  910--927. IEEE, 2020.

\bibitem{dong2023graph}
Guimin Dong, Mingyue Tang, Zhiyuan Wang, Jiechao Gao, Sikun Guo, Lihua Cai,
  Robert Gutierrez, Bradford Campbel, Laura~E Barnes, and Mehdi Boukhechba.
\newblock Graph neural networks in iot: a survey.
\newblock {\em ACM Transactions on Sensor Networks}, 19(2):1--50, 2023.

\bibitem{unpredictability}
Ben Edgington.
\newblock Lookahead.
\newblock \url{
  https://eth2book.info/capella/part2/building_blocks/randomness/\#lookahead},
  2023.

\bibitem{etherscan}
Etherscan.
\newblock Etherscan.
\newblock \url{https://etherscan.io/}, 2023.

\bibitem{gupta2023centralizing}
Tivas Gupta, Mallesh~M Pai, and Max Resnick.
\newblock The centralizing effects of private order flow on proposer-builder
  separation.
\newblock {\em arXiv preprint arXiv:2305.19150}, 2023.

\bibitem{heimbach2023ethereum}
Lioba Heimbach, Lucianna Kiffer, Christof Ferreira~Torres, and Roger
  Wattenhofer.
\newblock Ethereum's proposer-builder separation: Promises and realities.
\newblock In {\em Proceedings of the 2023 ACM on Internet Measurement
  Conference}, pages 406--420, 2023.

\bibitem{heimbach2022sok}
Lioba Heimbach and Roger Wattenhofer.
\newblock Sok: Preventing transaction reordering manipulations in decentralized
  finance.
\newblock {\em arXiv preprint arXiv:2203.11520}, 2022.

\bibitem{huang2021rich}
Yuming Huang, Jing Tang, Qianhao Cong, Andrew Lim, and Jianliang Xu.
\newblock Do the rich get richer? fairness analysis for blockchain incentives.
\newblock In {\em Proceedings of the 2021 International Conference on
  Management of Data}, pages 790--803, 2021.

\bibitem{jensen2023multi}
Johannes~Rude Jensen, Victor von Wachter, and Omri Ross.
\newblock Multi-block mev.
\newblock {\em arXiv preprint arXiv:2303.04430}, 2023.

\bibitem{kelkar2021themis}
Mahimna Kelkar, Soubhik Deb, Sishan Long, Ari Juels, and Sreeram Kannan.
\newblock Themis: Fast, strong order-fairness in byzantine consensus.
\newblock {\em Cryptology ePrint Archive}, 2021.

\bibitem{kelkar2020order}
Mahimna Kelkar, Fan Zhang, Steven Goldfeder, and Ari Juels.
\newblock Order-fairness for byzantine consensus.
\newblock In {\em Advances in Cryptology--CRYPTO 2020: 40th Annual
  International Cryptology Conference, CRYPTO 2020, Santa Barbara, CA, USA,
  August 17--21, 2020, Proceedings, Part III 40}, pages 451--480. Springer,
  2020.

\bibitem{kiayias2016blockchain}
Aggelos Kiayias, Elias Koutsoupias, Maria Kyropoulou, and Yiannis Tselekounis.
\newblock Blockchain mining games.
\newblock In {\em Proceedings of the 2016 ACM Conference on Economics and
  Computation}, pages 365--382, 2016.

\bibitem{kursawe2020wendy}
Klaus Kursawe.
\newblock Wendy, the good little fairness widget: Achieving order fairness for
  blockchains.
\newblock In {\em Proceedings of the 2nd ACM Conference on Advances in
  Financial Technologies}, pages 25--36, 2020.

\bibitem{li2018transaction}
Juanjuan Li, Yong Yuan, Shuai Wang, and Fei-Yue Wang.
\newblock Transaction queuing game in bitcoin blockchain.
\newblock In {\em 2018 IEEE Intelligent Vehicles Symposium (IV)}, pages
  114--119. IEEE, 2018.

\bibitem{liu2018evolutionary}
Xiaojun Liu, Wenbo Wang, Dusit Niyato, Narisa Zhao, and Ping Wang.
\newblock Evolutionary game for mining pool selection in blockchain networks.
\newblock {\em IEEE Wireless Communications Letters}, 7(5):760--763, 2018.

\bibitem{lyu2022empirical}
Xingyu Lyu, Mengya Zhang, Xiaokuan Zhang, Jianyu Niu, Yinqian Zhang, and
  Zhiqiang Lin.
\newblock An empirical study on ethereum private transactions and the security
  implications.
\newblock {\em arXiv preprint arXiv:2208.02858}, 2022.

\bibitem{mao2020perigee}
Yifan Mao, Soubhik Deb, Shaileshh~Bojja Venkatakrishnan, Sreeram Kannan, and
  Kannan Srinivasan.
\newblock Perigee: Efficient peer-to-peer network design for blockchains.
\newblock In {\em Proceedings of the 39th Symposium on Principles of
  Distributed Computing}, pages 428--437, 2020.

\bibitem{math11234741}
Yifan Mao and Shaileshh~Bojja Venkatakrishnan.
\newblock Less is more: Understanding network bias in proof-of-work
  blockchains.
\newblock {\em Mathematics}, 11(23), 2023.
\newblock URL: \url{https://www.mdpi.com/2227-7390/11/23/4741}, \href
  {http://dx.doi.org/10.3390/math11234741} {\path{doi:10.3390/math11234741}}.

\bibitem{mazorra2022price}
Bruno Mazorra, Michael Reynolds, and Vanesa Daza.
\newblock Price of mev: towards a game theoretical approach to mev.
\newblock In {\em Proceedings of the 2022 ACM CCS Workshop on Decentralized
  Finance and Security}, pages 15--22, 2022.

\bibitem{nakamoto2008bitcoin}
Satoshi Nakamoto.
\newblock Bitcoin: A peer-to-peer electronic cash system.
\newblock {\em Decentralized business review}, page 21260, 2008.

\bibitem{piet2022extracting}
Julien Piet, Jaiden Fairoze, and Nicholas Weaver.
\newblock Extracting godl [sic] from the salt mines: Ethereum miners extracting
  value.
\newblock {\em arXiv preprint arXiv:2203.15930}, 2022.

\bibitem{zeromev}
Pmcgoohan.
\newblock zeromev.
\newblock \url{https://www.zeromev.org/}, 2023.

\bibitem{qin2021quantifying}
Kaihua Qin, Liyi Zhou, and Arthur Gervais.
\newblock Quantifying blockchain extractable value: How dark is the forest?
\newblock {\em arXiv preprint arXiv:2101.05511}, 2021.

\bibitem{sahin2022optimal}
Hadi Sahin, Kemal Akkaya, and Sukumar Ganapati.
\newblock Optimal incentive mechanisms for fair and equitable rewards in pos
  blockchains.
\newblock In {\em 2022 IEEE International Performance, Computing, and
  Communications Conference (IPCCC)}, pages 367--373. IEEE, 2022.

\bibitem{vedula2023cobalt}
Arti Vedula, Abhishek Gupta, and Shaileshh~Bojja Venkatakrishnan.
\newblock Cobalt: Optimizing mining rewards in proof-of-work network games.
\newblock In {\em 2023 IEEE International Conference on Blockchain and
  Cryptocurrency (ICBC)}, pages 1--9. IEEE, 2023.

\bibitem{wahrstatter2023time}
Anton Wahrst{\"a}tter, Liyi Zhou, Kaihua Qin, Davor Svetinovic, and Arthur
  Gervais.
\newblock Time to bribe: Measuring block construction market.
\newblock {\em arXiv preprint arXiv:2305.16468}, 2023.

\bibitem{wang2022impact}
Ye~Wang, Patrick Zuest, Yaxing Yao, Zhicong Lu, and Roger Wattenhofer.
\newblock Impact and user perception of sandwich attacks in the defi ecosystem.
\newblock In {\em Proceedings of the 2022 CHI Conference on Human Factors in
  Computing Systems}, pages 1--15, 2022.

\bibitem{weintraub2022flash}
Ben Weintraub, Christof~Ferreira Torres, Cristina Nita-Rotaru, and Radu State.
\newblock A flash (bot) in the pan: Measuring maximal extractable value in
  private pools.
\newblock {\em arXiv preprint arXiv:2206.04185}, 2022.

\bibitem{xue2023goldfish}
Bowen Xue, Yifan Mao, Shaileshh~Bojja Venkatakrishnan, and Sreeram Kannan.
\newblock Goldfish: Peer selection using matrix completion in unstructured p2p
  network.
\newblock {\em arXiv preprint arXiv:2303.09761}, 2023.

\bibitem{zhang2023Kadabra}
Shaileshh Bojja~Venkatakrishnan Yunqi~Zhang.
\newblock Kadabra: Adapting kademlia for the decentralized web.
\newblock {\em arXiv preprint arXiv:2210.12858}, 2023.

\end{thebibliography}

\appendix
% \section{Figure of interactions between users, searchers, builders and proposers}
% \begin{figure}[htp]
%   \centering
%     \includegraphics[width=0.7\linewidth]{figures/mev_player.pdf}
%   \vspace{-1.2cm}
%   \caption{Interactions between users, searchers, builders and proposers} 
%   \label{fig:mevplayer}
% \end{figure}

\section{Validator reward when following Flashback policy}
\label{apx:vppolicy}
\begin{prop}[Validator reward when following Flashback policy]
\label{prop:validatorrewardpolicy}
\begin{align}
\mathbf{E}[V_p^\mathrm{policy}[t]] = (1-r_1) (\rho e^{-\rho/\mu_1} + \mu_1e^{-\rho/\mu_1}) + (1-r_2) e^{-\rho/\mu_1} 
(-\rho e^{-\rho/\mu_1} - \mu_1 e^{-\rho/\mu_1} + \mu_1 ) \notag \\
+ (1-r_2) \mu_3 e^{-\rho / \mu_1} + (1-e^{-\rho/\mu_1}) e^{-\rho / \mu_1}(1-r_2)\mu_2 + 
(1-e^{-\rho/\mu_1})(1-r_2)[ \mu_2 \notag \\ 
+ \frac{\mu_1}{(\mu_1 + \mu_2)} \rho e^{-\rho(\mu_1 + \mu_2)/(\mu_1 \mu_2) } 
+ \frac{\mu_1^2 \mu_2}{(\mu_1 + \mu_2)^2} e^{-\rho(\mu_1 + \mu_2)/(\mu_1 \mu_2) } - \frac{\mu_1^2 \mu_2}{(\mu_1 + \mu_2)^2} \notag \\
- \rho e^{-\rho/\mu_1} e^{-\rho/\mu_2} -  \mu_2 e^{-\rho/\mu_1} e^{-\rho/\mu_2}  ]
+ (1-e^{-\rho/\mu_1})e^{-\rho/\mu_1}(1-r_2) \mu_3 \notag \\
(1-e^{-\rho/\mu_1}) \mu_3 (1-r_2)
(1- \frac{\mu_1 }{\mu_1 + \mu_2} + \frac{\mu_1 }{\mu_1 + \mu_2} e^{- \rho(\mu_1 + \mu_2)/(\mu_1\mu_2)}  - e^{-\rho/\mu_1} e^{-\rho/\mu_2} ) \notag \\
+ (1 - r_2)(1 - e^{-\rho/\mu_1})( -(1 - e^{-\rho/\mu_2})\rho e^{-\rho/\mu_1}  - (1- e^{-\rho/\mu_2})\mu_1 e^{-\rho/\mu_1} -\frac{\mu_1}{\mu_1 + \mu_2}\rho e^{-\rho(\mu_1 + \mu_2)/(\mu_1\mu_2)} \notag \\
- \frac{\mu_1^2 \mu_2}{(\mu_1 + \mu_2)^2} e^{-
 \rho(\mu_1 + \mu_2)/(\mu_1 \mu_2)} + \frac{\mu_1^2 \mu_2}{(\mu_1 + \mu_2)^2}  
-\frac{\mu_1^2}{(\mu_1 + \mu_2)}  e^{-\rho(\mu_1 + \mu_2)/(\mu_1\mu_2)} + \frac{\mu_1^2}{(\mu_1 + \mu_2)}  ) \notag \\
+ (1-e^{-\rho/\mu_1})(1-r_2) \mu_3 ( \frac{\mu_1}{\mu_1 + \mu_2} (1 - e^{-\rho (\mu_1 + \mu_2)/(\mu_1 \mu_2)})  - e^{-\rho/\mu_1}(1-e^{-\rho/\mu_2})  ).  \label{eq:validatorrewpolicy}
\end{align}
\end{prop}

\begin{proof}
We have 
\begin{align}
\mathbf{E}[\mathbf{1}_{R[t] = 1}(1-r_1) X[t-1]] 
=
(1-r_1) \mathbf{E}[\mathbf{1}_{X[t-1] > \rho} X[t-1]] \notag \\
=
(1- r_1) \int_\rho^\infty x/\mu_1 e^{-x/\mu_1} dx = (1-r_1) [-xe^{-x/\mu_1} - \mu_1 e^{-x/\mu_1}]_\rho^\infty \notag \\
=
(1-r_1) (\rho e^{-\rho/\mu_1} + \mu_1e^{-\rho/\mu_1}) \\
\mathbf{E}[\mathbf{1}_{R[t] = 1}(1 - R[t+1]) (1-r_2) X[t]] = (1-r_2) \mathbf{E}[\mathbf{1}_{R[t]=1} \mathbf{1}_{R[t+1]=0}X[t]] \notag \\
= 
(1-r_2) \mathbf{E}[\mathbf{1}_{X[t-1]>\rho} \mathbf{1}_{X[t] < \rho} X[t]]
=
(1-r_2) e^{-\rho/\mu_1} \mathbf{E}[ \mathbf{1}_{X[t] < \rho} X[t]] \notag \\
= 
(1-r_2) e^{-\rho/\mu_1} 
\int_0^\rho x /\mu_1 e^{-x/\mu_1} dx 
=
(1-r_2) e^{-\rho/\mu_1} 
[-x e^{-x/\mu_1} - \mu_1 e^{-x/\mu_1}]_0^\rho \notag \\
=
(1-r_2) e^{-\rho/\mu_1} 
(-\rho e^{-\rho/\mu_1} - \mu_1 e^{-\rho/\mu_1} + \mu_1 ) \\
\mathbf{E}[\mathbf{1}_{R[t] = 1} (1-r_2) Y[t]] = (1-r_2) \mu_3 \mathbf{E}[\mathbf{1}_{X[t-1]>\rho}]
=
(1-r_2) \mu_3 e^{-\rho / \mu_1}
\end{align}
Next, 
\begin{align}
\mathbf{1}_{R[t]=0}\max \{  (1-r_2) X'[t] + (1-r_2) Y[t], (1-R[t+1])(1 - r_2)X[t] \notag \\
+ (1-r_2)Y[t] \} \notag \\
= 
\mathbf{1}_{R[t]=0} \mathbf{1}_{  (1-r_2) X'[t] + (1-r_2) Y[t] > (1-R[t+1])(1 - r_2)X[t] + (1-r_2)Y[t]} ((1-r_2) X'[t] \notag \\
+ (1-r_2) Y[t]) \notag \\
+ 
\mathbf{1}_{R[t]=0} \mathbf{1}_{  (1-r_2) X'[t] + (1-r_2) Y[t] < (1-R[t+1])(1 - r_2)X[t] + (1-r_2)Y[t]} ((1-R[t+1])(1 - r_2)X[t] \notag \\
+ (1-r_2)Y[t]). \notag 
\end{align}
Therefore,
\begin{align}
\mathbf{E}[\mathbf{1}_{R[t]=0} \mathbf{1}_{  (1-r_2) X'[t] + (1-r_2) Y[t] > (1-R[t+1])(1 - r_2)X[t] + (1-r_2)Y[t]} ((1-r_2) X'[t] + (1-r_2) Y[t])] \notag \\
=
\mathbf{E}[\mathbf{1}_{R[t]=0} \mathbf{1}_{  X'[t]  > (1-R[t+1])X[t] } ((1-r_2) X'[t] + (1-r_2) Y[t])] \notag \\
=
\mathbf{E}[\mathbf{1}_{R[t]=0} 
(\mathbf{1}_{R[t+1]=1} + \mathbf{1}_{R[t+1]=0} \mathbf{1}_{  X'[t]  > X[t] } ) ((1-r_2) X'[t] + (1-r_2) Y[t])] \notag \\
=
\mathbf{E}[\mathbf{1}_{R[t]=0} 
(\mathbf{1}_{R[t+1]=1}(1-r_2)X'[t] + \mathbf{1}_{R[t+1]=0} \mathbf{1}_{  X'[t]  > X[t] }(1-r_2) X'[t] + 
\mathbf{1}_{R[t+1]=1}(1-r_2)Y[t]  \notag \\
+ \mathbf{1}_{R[t+1]=0} \mathbf{1}_{  X'[t]  > X[t] }(1-r_2)Y[t]) ] \notag 
\end{align}
We evaluate each of the four terms in the above summation in order below.
\begin{align}
\mathbf{E}[\mathbf{1}_{R[t]=0}\mathbf{1}_{R[t+1]=1}(1-r_2)X'[t]] 
= 
\mathbf{E}[\mathbf{1}_{X[t-1]<\rho}\mathbf{1}_{X[t]>\rho }(1-r_2)X'[t]] \notag \\
=
(1-e^{-\rho/\mu_1}) e^{-\rho / \mu_1}(1-r_2)\mu_2 \label{eq:arbit1} \\
\mathbf{E}[\mathbf{1}_{R[t]=0}\mathbf{1}_{R[t+1]=0}\mathbf{1}_{X'[t]>X[t]}(1-r_2)X'[t]]
=
\mathbf{E}[\mathbf{1}_{X[t-1]<\rho}\mathbf{1}_{X[t]<\rho}\mathbf{1}_{X'[t]>X[t]}(1-r_2)X'[t]] \notag \\
=
(1-e^{-\rho/\mu_1})
\mathbf{E}[\mathbf{1}_{X[t]<\rho}\mathbf{1}_{X'[t]>X[t]}(1-r_2)X'[t]] \notag 
\end{align}
\begin{align}
=
(1-e^{-\rho/\mu_1})
\mathbf{E}[ \mathbf{E}[ \mathbf{1}_{X[t]<\rho}\mathbf{1}_{X'[t]>X[t]}(1-r_2)X'[t] | X'[t]]] \notag \\
= 
(1-e^{-\rho/\mu_1})
\mathbf{E}[ (1-r_2)X'[t] \mathbf{E}[ \mathbf{1}_{X[t]<\rho}\mathbf{1}_{X'[t]>X[t]} | X'[t]]] \notag \\
=
(1-e^{-\rho/\mu_1})
\mathbf{E}[ (1-r_2)X'[t] (1-e^{-\min(X'[t],\rho)/\mu_1})] \notag \\
=
(1-e^{-\rho/\mu_1})(1-r_2) \int_0^\infty x (1-e^{-\min(x,\rho)/\mu_1}) /\mu_2 e^{-x/\mu_2} dx \notag 
\end{align}
\begin{align}
=
(1-e^{-\rho/\mu_1})(1-r_2)[\int_0^\infty x/\mu_2 e^{-x/\mu_2} dx - \int_0^\infty x/\mu_2 e^{-\min(x,\rho)/\mu_1}e^{-x/\mu_2} dx ] \notag \\
=
(1-e^{-\rho/\mu_1})(1-r_2)[ [-x e^{-x/\mu_2} - \mu_2 e^{-x/\mu_2}]_0^\infty - \int_0^\rho x/\mu_2 e^{-x/\mu_1}e^{-x/\mu_2} dx \notag \\
- \int_\rho^\infty x/\mu_2 e^{-\rho/\mu_1} e^{-x/\mu_2} dx] \notag \\
=
(1-e^{-\rho/\mu_1})(1-r_2)[ \mu_2 - \int_0^\rho x/\mu_2 e^{-x(\mu_1 + \mu_2)/(\mu_1 \mu_2) } dx - \frac{e^{-\rho/\mu_1}}{\mu_2} \int_\rho^\infty x  e^{-x/\mu_2} dx] \notag 
\end{align}
\begin{align}
= 
(1-e^{-\rho/\mu_1})(1-r_2)[ \mu_2 - [ -\frac{\mu_1}{(\mu_1 + \mu_2)} x e^{-x(\mu_1 + \mu_2)/(\mu_1 \mu_2) } - \frac{\mu_1^2 \mu_2}{(\mu_1 + \mu_2)^2} e^{-x(\mu_1 + \mu_2)/(\mu_1 \mu_2) }]_0^\rho \notag \\
- \frac{e^{-\rho/\mu_1}}{\mu_2} [- \mu_2 x  e^{-x/\mu_2} - \mu_2^2 e^{-x/\mu_2} ]_\rho^\infty ] \notag \\
= 
(1-e^{-\rho/\mu_1})(1-r_2)[ \mu_2 + \frac{\mu_1}{(\mu_1 + \mu_2)} \rho e^{-\rho(\mu_1 + \mu_2)/(\mu_1 \mu_2) } + \frac{\mu_1^2 \mu_2}{(\mu_1 + \mu_2)^2} e^{-\rho(\mu_1 + \mu_2)/(\mu_1 \mu_2) } - \frac{\mu_1^2 \mu_2}{(\mu_1 + \mu_2)^2} \notag \\
- \frac{e^{-\rho/\mu_1}}{\mu_2} [\mu_2 \rho  e^{-\rho/\mu_2} + \mu_2^2 e^{-\rho/\mu_2} ] ] \notag 
\end{align}
\begin{align}
=  
(1-e^{-\rho/\mu_1})(1-r_2)[ \mu_2 + \frac{\mu_1}{(\mu_1 + \mu_2)} \rho e^{-\rho(\mu_1 + \mu_2)/(\mu_1 \mu_2) } + \frac{\mu_1^2 \mu_2}{(\mu_1 + \mu_2)^2} e^{-\rho(\mu_1 + \mu_2)/(\mu_1 \mu_2) } \notag \\ - \frac{\mu_1^2 \mu_2}{(\mu_1 + \mu_2)^2} 
- \rho e^{-\rho/\mu_1} e^{-\rho/\mu_2} -  \mu_2 e^{-\rho/\mu_1} e^{-\rho/\mu_2}  ] \\
\mathbf{E}[\mathbf{1}_{R[t]=0}\mathbf{1}_{R[t+1]=1}(1-r_2)Y[t]]
=
\mathbf{E}[\mathbf{1}_{X[t-1]<\rho}\mathbf{1}_{X[t]>\rho}(1-r_2)Y[t]]
= 
(1-e^{-\rho/\mu_1})e^{-\rho/\mu_1}(1-r_2) \mu_3 \\
\mathbf{E}[\mathbf{1}_{R[t]=0}\mathbf{1}_{R[t+1]=0}\mathbf{1}_{X'[t]>X[t]}(1-r_2)Y[t]]
=
\mathbf{E}[\mathbf{1}_{X[t-1]<\rho}\mathbf{1}_{X[t]<\rho}\mathbf{1}_{X'[t]>X[t]}(1-r_2)Y[t]] \notag 
\end{align}
\begin{align}
= 
(1-e^{-\rho/\mu_1}) \mu_3 (1-r_2)
\mathbf{E}[\mathbf{1}_{X[t]<\rho}\mathbf{1}_{X'[t]>X[t]}] \notag \\
=
(1-e^{-\rho/\mu_1}) \mu_3 (1-r_2)
\mathbf{E}[ \mathbf{E}[ \mathbf{1}_{X[t]<\min(\rho, X'[t])} | X'[t]]] \notag \\
=
(1-e^{-\rho/\mu_1}) \mu_3 (1-r_2)
\mathbf{E}[1-e^{-\min(\rho, X'[t])/\mu_1} ] \notag \\
=
(1-e^{-\rho/\mu_1}) \mu_3 (1-r_2)
(1-\int_0^\infty e^{-\min(\rho, x)/\mu_1}/\mu_2 e^{-x/\mu_2} dx ) \notag 
\end{align}
\begin{align}
= 
(1-e^{-\rho/\mu_1}) \mu_3 (1-r_2)
(1-\int_0^\rho e^{- x/\mu_1}/\mu_2 e^{-x/\mu_2} dx - \int_\rho^\infty e^{-\rho/\mu_1}/\mu_2 e^{-x/\mu_2} dx ) \notag \\
=
(1-e^{-\rho/\mu_1}) \mu_3 (1-r_2)
(1- \frac{1}{\mu_2} \int_0^\rho e^{- x(\mu_1 + \mu_2)/(\mu_1\mu_2)}  dx - e^{-\rho/\mu_1} e^{-\rho/\mu_2} ) \notag \\
=
(1-e^{-\rho/\mu_1}) \mu_3 (1-r_2)
(1- \frac{1}{\mu_2} [-\frac{\mu_1 \mu_2}{\mu_1 + \mu_2} e^{- x(\mu_1 + \mu_2)/(\mu_1\mu_2)}]_0^\rho  - e^{-\rho/\mu_1} e^{-\rho/\mu_2} ) \notag \\
=
(1-e^{-\rho/\mu_1}) \mu_3 (1-r_2)
(1- \frac{1}{\mu_2} (\frac{\mu_1 \mu_2}{\mu_1 + \mu_2} -\frac{\mu_1 \mu_2}{\mu_1 + \mu_2} e^{- \rho(\mu_1 + \mu_2)/(\mu_1\mu_2)})  - e^{-\rho/\mu_1} e^{-\rho/\mu_2} ) \notag \\
=
(1-e^{-\rho/\mu_1}) \mu_3 (1-r_2)
(1- \frac{\mu_1 }{\mu_1 + \mu_2} + \frac{\mu_1 }{\mu_1 + \mu_2} e^{- \rho(\mu_1 + \mu_2)/(\mu_1\mu_2)}  - e^{-\rho/\mu_1} e^{-\rho/\mu_2} )
\end{align}
Next we have 
\begin{align}
\mathbf{E}[\mathbf{1}_{R[t]=0} \mathbf{1}_{  (1-r_2) X'[t] + (1-r_2) Y[t] < (1-R[t+1])(1 - r_2)X[t] + (1-r_2)Y[t]} ((1-R[t+1])(1 - r_2)X[t] \notag \\
+ (1-r_2)Y[t])] \notag \\
=
\mathbf{E}[\mathbf{1}_{R[t]=0} \mathbf{1}_{  X'[t] < (1-R[t+1])X[t] } ((1-R[t+1])(1 - r_2)X[t] + (1-r_2)Y[t])] \notag \\
=
\mathbf{E}[\mathbf{1}_{R[t]=0} 
\mathbf{1}_{R[t+1]=0}\mathbf{1}_{  X'[t] < X[t] } ((1 - r_2)X[t] + (1-r_2)Y[t])]. \notag 
\end{align}
We evaluate each of the two terms in that above summation below. 
\begin{align}
\mathbf{E}[\mathbf{1}_{R[t]=0} 
\mathbf{1}_{R[t+1]=0}\mathbf{1}_{  X'[t] < X[t] } (1 - r_2)X[t]] 
=
\mathbf{E}[\mathbf{1}_{X[t-1]<\rho} 
\mathbf{1}_{X[t]<\rho}\mathbf{1}_{  X'[t] < X[t] } (1 - r_2)X[t]] \notag \\
=
(1 - r_2)(1 - e^{-\rho/\mu_1}) \mathbf{E}[ 
\mathbf{1}_{X[t]<\rho}\mathbf{1}_{  X'[t] < X[t] } X[t]] \notag \\
=
(1 - r_2)(1 - e^{-\rho/\mu_1}) \mathbf{E}[\mathbf{1}_{X'[t] < \rho} \mathbf{1}_{  X'[t] < X[t] < \rho } X[t]] \notag \\
=
(1 - r_2)(1 - e^{-\rho/\mu_1}) \mathbf{E}[\mathbf{E}[\mathbf{1}_{X'[t] < \rho} \mathbf{1}_{  X'[t] < X[t] < \rho } X[t] | X'[t]]] \notag \\
=
(1 - r_2)(1 - e^{-\rho/\mu_1}) \mathbf{E}[\mathbf{1}_{X'[t] < \rho} \int_{X'[t]}^\rho x / \mu_1 e^{-x/\mu_1} dx  ] \notag \\
=
(1 - r_2)(1 - e^{-\rho/\mu_1}) \mathbf{E}[\mathbf{1}_{X'[t] < \rho} [ -x e^{-x/\mu_1} - \mu_1 e^{-x/\mu_1}]_{X'[t]}^\rho  ] \notag \\
=
(1 - r_2)(1 - e^{-\rho/\mu_1}) \mathbf{E}[\mathbf{1}_{X'[t] < \rho} [ -\rho e^{-\rho/\mu_1} - \mu_1 e^{-\rho/\mu_1} + X'[t] e^{-X'[t]/\mu_1} + \mu_1 e^{-X'[t]/\mu_1}]  ] \notag \\
=
(1 - r_2)(1 - e^{-\rho/\mu_1}) \mathbf{E}[ -\mathbf{1}_{X'[t] < \rho} \rho e^{-\rho/\mu_1} - \mathbf{1}_{X'[t] < \rho} \mu_1 e^{-\rho/\mu_1} + \mathbf{1}_{X'[t] < \rho} X'[t] e^{-X'[t]/\mu_1} \notag \\
+ \mathbf{1}_{X'[t] < \rho} \mu_1 e^{-X'[t]/\mu_1}]  \notag \\
=
(1 - r_2)(1 - e^{-\rho/\mu_1})( -(1 - e^{-\rho/\mu_2})\rho e^{-\rho/\mu_1}  - (1- e^{-\rho/\mu_2})\mu_1 e^{-\rho/\mu_1} + \int_0^\rho x e^{-x/\mu_1}/\mu_2 e^{-x/\mu_2} dx \notag \\
+\int_0^\rho  \mu_1 e^{-x/\mu_1} /\mu_2 e^{-x/\mu_2}dx) \notag \\
= 
(1 - r_2)(1 - e^{-\rho/\mu_1})( -(1 - e^{-\rho/\mu_2})\rho e^{-\rho/\mu_1}  - (1- e^{-\rho/\mu_2})\mu_1 e^{-\rho/\mu_1} + \int_0^\rho x e^{-x(\mu_1 + \mu_2)/(\mu_1\mu_2)}/\mu_2 dx \notag \\
+\int_0^\rho  \mu_1 e^{-x(\mu_1 + \mu_2)/(\mu_1\mu_2)} /\mu_2 dx) \notag \\
= 
(1 - r_2)(1 - e^{-\rho/\mu_1})( -(1 - e^{-\rho/\mu_2})\rho e^{-\rho/\mu_1}  - (1- e^{-\rho/\mu_2})\mu_1 e^{-\rho/\mu_1} + [-\frac{\mu_1}{\mu_1 + \mu_2}x e^{-x(\mu_1 + \mu_2)/(\mu_1\mu_2)} \notag \\
 - \frac{\mu_1^2 \mu_2}{(\mu_1 + \mu_2)^2} e^{-x(\mu_1 + \mu_2)/(\mu_1 \mu_2)}]_0^\rho 
+ [-\frac{\mu_1^2}{(\mu_1 + \mu_2)}  e^{-x(\mu_1 + \mu_2)/(\mu_1\mu_2)}  ]_0^\rho) \notag \\
=
(1 - r_2)(1 - e^{-\rho/\mu_1})( -(1 - e^{-\rho/\mu_2})\rho e^{-\rho/\mu_1}  - (1- e^{-\rho/\mu_2})\mu_1 e^{-\rho/\mu_1} -\frac{\mu_1}{\mu_1 + \mu_2}\rho e^{-\rho(\mu_1 + \mu_2)/(\mu_1\mu_2)} \notag \\
 - \frac{\mu_1^2 \mu_2}{(\mu_1 + \mu_2)^2} e^{-
 \rho(\mu_1 + \mu_2)/(\mu_1 \mu_2)} + \frac{\mu_1^2 \mu_2}{(\mu_1 + \mu_2)^2}  
-\frac{\mu_1^2}{(\mu_1 + \mu_2)}  e^{-\rho(\mu_1 + \mu_2)/(\mu_1\mu_2)} + \frac{\mu_1^2}{(\mu_1 + \mu_2)}  )
\end{align}
Next, 
\begin{align}
\mathbf{E}[\mathbf{1}_{R[t]=0}\mathbf{1}_{R[t+1]=0}\mathbf{1}_{X'[t] < X[t]}(1-r_2)Y[t]]
=
\mathbf{E}[\mathbf{1}_{X[t-1]<\rho}\mathbf{1}_{X[t]<\rho }\mathbf{1}_{X'[t] < X[t]}(1-r_2)Y[t]] \notag \\
=
(1-e^{-\rho/\mu_1})(1-r_2) \mu_3
\mathbf{E}[\mathbf{1}_{X[t]<\rho }\mathbf{1}_{X'[t] < X[t]}] \notag \\
=
(1-e^{-\rho/\mu_1})(1-r_2) \mu_3
\mathbf{E}[\mathbf{1}_{X'[t] < \rho }\mathbf{1}_{X'[t] < X[t]<\rho }] \notag \\
=
(1-e^{-\rho/\mu_1})(1-r_2) \mu_3
\mathbf{E}[\mathbf{E}[\mathbf{1}_{X'[t] < \rho }\mathbf{1}_{X'[t] < X[t]<\rho } | X'[t]]] \notag \\
=
(1-e^{-\rho/\mu_1})(1-r_2) \mu_3
\mathbf{E}[\mathbf{1}_{X'[t] < \rho } (e^{-X'[t]/\mu_1} - e^{-\rho/\mu_1}) ] \notag \\
=
(1-e^{-\rho/\mu_1})(1-r_2) \mu_3 \int_0^\rho 
 (e^{-x/\mu_1} - e^{-\rho/\mu_1})/\mu_2 e^{-x/\mu_2} dx \notag \\
 =
 (1-e^{-\rho/\mu_1})(1-r_2) \mu_3 ( \int_0^\rho 
 e^{-x(\mu_1 + \mu_2)/(\mu_1\mu_2)}/\mu_2 dx - \int_0^\rho e^{-\rho/\mu_1}/\mu_2 e^{-x/\mu_2} dx ) \notag \\
 =
 (1-e^{-\rho/\mu_1})(1-r_2) \mu_3 ( \frac{\mu_1}{\mu_1 + \mu_2} \int_0^\rho 
 \frac{\mu_1 + \mu_2}{\mu_1 \mu_2} e^{-x(\mu_1 + \mu_2)/(\mu_1\mu_2)} dx - e^{-\rho/\mu_1}(1-e^{-\rho/\mu_2})  ) \notag \\
 =
 (1-e^{-\rho/\mu_1})(1-r_2) \mu_3 ( \frac{\mu_1}{\mu_1 + \mu_2} (1 - e^{-\rho (\mu_1 + \mu_2)/(\mu_1 \mu_2)})  - e^{-\rho/\mu_1}(1-e^{-\rho/\mu_2})  ).  \label{eq:arbit2}
\end{align}
\end{proof}

\section{Validator reward when following default policy}
\label{apx:vpdefault}
\begin{prop}[Validator reward when following default policy]
\label{prop:validatordeviation}
\begin{align}
\mathbf{E}[V_p^\mathrm{default}[t]] = (
(1-e^{-\rho/\mu_1}) e^{-\rho / \mu_1}(1-r_2)\mu_2 + 
(1-e^{-\rho/\mu_1})(1-r_2)[ \mu_2 + \frac{\mu_1}{(\mu_1 + \mu_2)} \rho e^{-\rho(\mu_1 + \mu_2)/(\mu_1 \mu_2) } \notag \\
+ \frac{\mu_1^2 \mu_2}{(\mu_1 + \mu_2)^2} e^{-\rho(\mu_1 + \mu_2)/(\mu_1 \mu_2) } - \frac{\mu_1^2 \mu_2}{(\mu_1 + \mu_2)^2} 
- \rho e^{-\rho/\mu_1} e^{-\rho/\mu_2} -  \mu_2 e^{-\rho/\mu_1} e^{-\rho/\mu_2}  ] 
+ (1-e^{-\rho/\mu_1}) \notag \\ 
e^{-\rho/\mu_1}(1-r_2) \mu_3 
(1-e^{-\rho/\mu_1}) \mu_3 (1-r_2)
(1- \frac{\mu_1 }{\mu_1 + \mu_2} + \frac{\mu_1 }{\mu_1 + \mu_2} e^{- \rho(\mu_1 + \mu_2)/(\mu_1\mu_2)}  - e^{-\rho/\mu_1} e^{-\rho/\mu_2} ) \notag \\
+ (1 - r_2)(1 - e^{-\rho/\mu_1})( -(1 - e^{-\rho/\mu_2})\rho e^{-\rho/\mu_1}  - (1- e^{-\rho/\mu_2})\mu_1 e^{-\rho/\mu_1} -\frac{\mu_1}{\mu_1 + \mu_2}\rho e^{-\rho(\mu_1 + \mu_2)/(\mu_1\mu_2)} \notag \\
- \frac{\mu_1^2 \mu_2}{(\mu_1 + \mu_2)^2} e^{-
 \rho(\mu_1 + \mu_2)/(\mu_1 \mu_2)} + \frac{\mu_1^2 \mu_2}{(\mu_1 + \mu_2)^2}  
-\frac{\mu_1^2}{(\mu_1 + \mu_2)}  e^{-\rho(\mu_1 + \mu_2)/(\mu_1\mu_2)} + \frac{\mu_1^2}{(\mu_1 + \mu_2)}  ) \notag \\
+ (1-e^{-\rho/\mu_1})(1-r_2) \mu_3 ( \frac{\mu_1}{\mu_1 + \mu_2} (1 - e^{-\rho (\mu_1 + \mu_2)/(\mu_1 \mu_2)})  - e^{-\rho/\mu_1}(1-e^{-\rho/\mu_2})  ))/(1-e^{-\rho/\mu_1}).  \label{eq:validrewdeviation}
\end{align}
\end{prop}

\begin{proof}
To evaluate $V_p^\mathrm{default}[t]$ we note that   
\begin{align}
\mathbf{E}[V_p^\mathrm{default}[t]] = \mathbf{E}[ \mathbf{1}_{R[t]=0} \max \{  (1-r_2) X'[t] + (1-r_2) Y[t], (1-R[t+1])(1 - r_2)X[t] \notag \\
+ (1-r_2)Y[t] \}]/(1 - e^{-\rho/\mu_1}). \notag  
\end{align}
However, the numerator of the above has previously been evaluated in equations~\eqref{eq:arbit1}--\eqref{eq:arbit2}. 
\end{proof}

\section{Primary builder reward}
\label{apx:vprimary}
\begin{prop}[Primary builder reward]
\label{prop:primarybuilder}
\begin{align}
\mathbf{E}[V_\mathrm{primary}[t]] = r_1 (\rho e^{-\rho/\mu_1} + \mu_1 e^{-\rho/\mu_1})  
+ r_2 e^{-\rho / \mu_1} (\mu_1 -\rho e^{-\rho/\mu_1} - \mu_1 e^{-\rho / \mu_1})
+ r_2 \mu_3 e^{-\rho / \mu_1} \notag \\
+ r_2 (1 - e^{-\rho / \mu_1})  [ -\rho e^{-\rho/\mu_1}(1-e^{-\rho/\mu_2}) - \mu_1 e^{-\rho/\mu_1}(1-e^{-\rho/\mu_2}) 
-\frac{\rho \mu_1}{\mu_1 + \mu_2} e^{-\rho(\mu_1 + \mu_2)/(\mu_1 \mu_2)} \notag \\
- \frac{\mu_1^2 \mu_2}{(\mu_1 + \mu_2)^2} e^{-\rho(\mu_1 + \mu_2)/(\mu_1 \mu_2)} +  \frac{\mu_1^2 \mu_2}{(\mu_1 + \mu_2)^2}  +  \frac{\mu_1^2}{(\mu_1 + \mu_2)}(1 - e^{-\rho (\mu_1 + \mu_2)/(\mu_1 \mu_2)})    ] \notag \\
+ r_2 \mu_3 (1 - e^{-\rho/\mu_1})\left(  \frac{\mu_1}{(\mu_1 + \mu_2)} (1 - e^{-\rho(\mu_1 + \mu_2)/(\mu_1 \mu_2)})   - e^{-\rho/\mu_1}(1 - e^{-\rho/\mu_2}) \right). 
\end{align}
\end{prop}

\begin{proof}
\begin{align}
\mathbf{E}[\mathbf{1}_{R[t] = 1} r_1 X[t-1]] = r_1 \mathbf{E}[\mathbf{1}_{X[t-1] > \rho} X[t-1]] = r_1 \int_\rho^\infty x/\mu_1 e^{-x/\mu_1} dx = r_1 [ -xe^{-x/\mu_1} - \mu_1 e^{-x/\mu_1}]_\rho ^\infty \notag \\
= r_1 (\rho e^{-\rho/\mu_1} + \mu_1 e^{-\rho/\mu_1}) \label{eq:vprim1} \\
\mathbf{E}[\mathbf{1}_{R[t] = 1}(1 - R[t+1])r_2 X[t]] = \mathbf{E}[\mathbf{1}_{X[t-1] > \rho} \mathbf{1}_{X[t] < \rho} r_2 X[t]] = r_2 e^{-\rho / \mu_1} \int_0 ^\rho x e^{-x/\mu_1} / \mu_1 dx \notag \\
= r_2 e^{-\rho / \mu_1} [-xe^{-x/\mu_1} - \mu_1 e^{-x/\mu_1}]_0^\rho =  r_2 e^{-\rho / \mu_1} (\mu_1 -\rho e^{-\rho/\mu_1} - \mu_1 e^{-\rho / \mu_1}) \\
\mathbf{E}[\mathbf{1}_{R[t]=1} r_2 Y[t]] = r_2 \mu_3 P(X[t-1] > \rho) = r_2 \mu_3 e^{-\rho / \mu_1}  \\
\mathbf{E} [ \mathbf{1}_{R[t]=0} \mathbf{1}_{(1-r_2) X'[t] + (1-r_2) Y[t] < (1-R[t+1])(1 - r_2)X[t] + (1-r_2)Y[t]} (1-R[t+1]) r_2 X[t] ] \notag 
\end{align}
\begin{align}
= \mathbf{E} [ \mathbf{1}_{R[t]=0} \mathbf{1}_{X'[t]  < (1-R[t+1])X[t] } (1-R[t+1]) r_2 X[t] ] 
= \mathbf{E} [ \mathbf{1}_{R[t]=0} \mathbf{1}_{R[t+1]=0} \mathbf{1}_{X'[t]  < X[t] }  r_2 X[t] ]  \notag \\
= r_2 (1 - e^{-\rho / \mu_1}) \mathbf{E} [ \mathbf{1}_{R[t+1]=0} \mathbf{1}_{X'[t]  < X[t] }  X[t] ] 
= r_2 (1 - e^{-\rho / \mu_1}) \mathbf{E} [ \mathbf{E} [ \mathbf{1}_{R[t+1]=0} \mathbf{1}_{X'[t]  < X[t] }  X[t] ] | X'[t] ] \notag \\
= r_2 (1 - e^{-\rho / \mu_1}) \mathbf{E} [ \mathbf{E} [ \mathbf{1}_{X[t] < \rho} \mathbf{1}_{X'[t]  < X[t] }  X[t] ] | X'[t] ]  \notag \\
= r_2 (1 - e^{-\rho / \mu_1}) \mathbf{E} [ \mathbf{E} [\mathbf{1}_{X'[t] < \rho} \mathbf{1}_{X'[t] < X[t] < \rho} X[t] ] | X'[t] ]  \notag 
\end{align}
\begin{align}
= r_2 (1 - e^{-\rho / \mu_1}) \mathbf{E} [ \mathbf{1}_{X'[t] < \rho} \mathbf{E} [ \mathbf{1}_{X'[t] < X[t] < \rho} X[t] ] | X'[t] ]  \notag \\ 
= r_2 (1 - e^{-\rho / \mu_1}) \mathbf{E} [ \mathbf{1}_{X'[t] < \rho} \int_{X'[t]}^\rho x /\mu_1e^{-x/\mu_1} dx ] \notag  \\ 
= r_2 (1 - e^{-\rho / \mu_1}) \mathbf{E} [ \mathbf{1}_{X'[t] < \rho} [ -x e^{-x/\mu_1} - \mu_1 e^{-x/\mu_1}]^\rho_{X'[t]}  ] \notag \\ 
= r_2 (1 - e^{-\rho / \mu_1}) \mathbf{E} [ \mathbf{1}_{X'[t] < \rho} [ -\rho e^{-\rho/\mu_1} - \mu_1 e^{-\rho/\mu_1} + X'[t] e^{-X'[t]/\mu_1} + \mu_1 e^{-X'[t]/\mu_1}]  ] \notag \\ 
= r_2 (1 - e^{-\rho / \mu_1})  [ -\rho e^{-\rho/\mu_1}(1-e^{-\rho/\mu_2}) - \mu_1 e^{-\rho/\mu_1}(1-e^{-\rho/\mu_2}) \notag \\
+ \int_0^\rho  x/\mu_2 e^{-x/\mu_1} e^{-x/\mu_2} dx + \int_0^\rho \mu_1/\mu_2 e^{-x/\mu_1}e^{-x/\mu_2} dx   ] \notag 
\end{align}
\begin{align}
= r_2 (1 - e^{-\rho / \mu_1})  [ -\rho e^{-\rho/\mu_1}(1-e^{-\rho/\mu_2}) - \mu_1 e^{-\rho/\mu_1}(1-e^{-\rho/\mu_2}) \notag \\
+ \int_0^\rho  x/\mu_2 e^{-x(\mu_1 + \mu_2)/(\mu_1 \mu_2)} dx + \int_0^\rho \mu_1/\mu_2 e^{-x(\mu_1 + \mu_2)/(\mu_1 \mu_2)}  dx   ] \notag \\ 
= r_2 (1 - e^{-\rho / \mu_1})  [ -\rho e^{-\rho/\mu_1}(1-e^{-\rho/\mu_2}) - \mu_1 e^{-\rho/\mu_1}(1-e^{-\rho/\mu_2}) \notag \\
+  \left[-\frac{x \mu_1}{\mu_1 + \mu_2} e^{-x(\mu_1 + \mu_2)/(\mu_1 \mu_2)} - \frac{\mu_1^2 \mu_2}{(\mu_1 + \mu_2)^2} e^{-x(\mu_1 + \mu_2)/(\mu_1 \mu_2)} \right]_0^\rho +  \frac{\mu_1^2}{(\mu_1 + \mu_2)}(1 - e^{-\rho (\mu_1 + \mu_2)/(\mu_1 \mu_2)})    ] \notag \\ 
= r_2 (1 - e^{-\rho / \mu_1})  [ -\rho e^{-\rho/\mu_1}(1-e^{-\rho/\mu_2}) - \mu_1 e^{-\rho/\mu_1}(1-e^{-\rho/\mu_2}) 
 -\frac{\rho \mu_1}{\mu_1 + \mu_2} e^{-\rho(\mu_1 + \mu_2)/(\mu_1 \mu_2)} \notag \\ - \frac{\mu_1^2 \mu_2}{(\mu_1 + \mu_2)^2} e^{-\rho(\mu_1 + \mu_2)/(\mu_1 \mu_2)} +  \frac{\mu_1^2 \mu_2}{(\mu_1 + \mu_2)^2}  +  \frac{\mu_1^2}{(\mu_1 + \mu_2)}(1 - e^{-\rho (\mu_1 + \mu_2)/(\mu_1 \mu_2)})    ] \\ 
\mathbf{E} [ \mathbf{1}_{R[t]=0} \mathbf{1}_{(1-r_2) X'[t] + (1-r_2) Y[t] < (1-R[t+1])(1 - r_2)X[t] + (1-r_2)Y[t]}  r_2 Y[t] ] \notag \\
= \mathbf{E} [ \mathbf{1}_{R[t]=0} \mathbf{1}_{R[t+1] = 0} \mathbf{1}_{ X'[t]  < X[t] }  r_2 Y[t] ] 
= \mathbf{E} [ \mathbf{1}_{X[t-1] < \rho} \mathbf{1}_{X[t] < \rho} \mathbf{1}_{ X'[t]  < X[t] }  r_2 Y[t] ]  \notag \\
= r_2 \mu_3 (1 - e^{-\rho/\mu_1}) \mathbf{E} [ \mathbf{1}_{X[t] < \rho}  \mathbf{1}_{ X'[t]  < X[t] }   ] 
= r_2 \mu_3 (1 - e^{-\rho/\mu_1}) \mathbf{E} [ \mathbf{E}[ \mathbf{1}_{X[t] < \rho}  \mathbf{1}_{ X'[t]  < X[t] }  | X'[t] ] ] \notag 
\end{align}
\begin{align}
= r_2 \mu_3 (1 - e^{-\rho/\mu_1}) \mathbf{E} [ \mathbf{E}[  \mathbf{1}_{ X'[t]  < X[t] < \rho } \mathbf{1}_{X'[t] < \rho} | X'[t] ] ] \notag \\
= r_2 \mu_3 (1 - e^{-\rho/\mu_1}) \mathbf{E} [ \mathbf{1}_{X'[t] < \rho} \mathbf{E}[  \mathbf{1}_{ X'[t]  < X[t] < \rho }  | X'[t] ] ] \notag \\
= r_2 \mu_3 (1 - e^{-\rho/\mu_1}) \mathbf{E} [ \mathbf{1}_{X'[t] < \rho} (e^{-X'[t]/\mu_1} - e^{-\rho/\mu_1})  ] \notag 
\end{align}
\begin{align}
= r_2 \mu_3 (1 - e^{-\rho/\mu_1})\left( \int_0^\rho 1/\mu_2 e^{-x/\mu_1} e^{-x/\mu_2} dx  - e^{-\rho/\mu_1}(1 - e^{-\rho/\mu_2}) \right) \notag \\
= r_2 \mu_3 (1 - e^{-\rho/\mu_1})\left( \int_0^\rho (\mu_1/(\mu_1 + \mu_2)) (\mu_1 + \mu_2)/(\mu_1 \mu_2) e^{-x(\mu_1 + \mu_2)/(\mu_1 \mu_2)}  dx  - e^{-\rho/\mu_1}(1 - e^{-\rho/\mu_2}) \right) \notag \\
= r_2 \mu_3 (1 - e^{-\rho/\mu_1})\left(  \frac{\mu_1}{(\mu_1 + \mu_2)} (1 - e^{-\rho(\mu_1 + \mu_2)/(\mu_1 \mu_2)})   - e^{-\rho/\mu_1}(1 - e^{-\rho/\mu_2}) \right) \label{eq:vprimlast}
\end{align}  
\end{proof}

\section{Secondary builder reward}
\label{apx:vsecondary}
\begin{prop}[Secondary builder reward]
\label{prop:secondarybuilder}
\begin{align}
\mathbf{E}[V_\mathrm{secondary}[t]] = r_2 (1 - e^{-\rho/\mu_1}) (e^{-\rho / \mu_1} \mu_2 + \mu_2 + \rho \mu_1    /(\mu_1 + \mu_2) e^{- \rho (\mu_1 + \mu_2)/(\mu_1 \mu_2)} \notag \\
+  \frac{\mu_1^2 \mu_2}{(\mu_1 + \mu_2)^2} e^{-\rho(\mu_1 + \mu_2)/(\mu_1 \mu_2)}
- \frac{\mu_1^2 \mu_2}{(\mu_1 + \mu_2)^2}
 - \rho e^{- \rho (\mu_1 + \mu_2) /(\mu_1 \mu_2)} - \mu_2 e^{-\rho(\mu_1 + \mu_2)/(\mu_1 \mu_2)} ) \notag \\
 + r_2 \mu_3 (1 - e^{-\rho/\mu_1}) ( e^{-\rho/\mu_1} +  1 + \mu_1/(\mu_1 + \mu_2) e^{-\rho (\mu_1 + \mu_2)/(\mu_1 \mu_2)} - \mu_1/(\mu_1 + \mu_2)   -  e^{-\rho/\mu_1} e^{-\rho/\mu_2}  ) .
\end{align}
\end{prop}

\begin{proof}
\begin{align}
\mathbf{E}[\mathbf{1}_{R[t]=0} \mathbf{1}_{(1-r_2) X'[t] + (1-r_2) Y[t] > (1-R[t+1])(1 - r_2)X[t] + (1-r_2)Y[t]} r_2 X'[t] ] \notag \\
= \mathbf{E} [ \mathbf{1}_{R[t]=0} \mathbf{1}_{ X'[t] > (1-R[t+1])X[t]} r_2 X'[t] ] 
= 
\mathbf{E} [ \mathbf{1}_{R[t]=0} \mathbf{1}_{ X'[t] > (1-R[t+1])X[t]} r_2 X'[t] ] \notag \\ 
= 
r_2 (1 - e^{-\rho/\mu_1}) \mathbf{E} [  \mathbf{1}_{ X'[t] > (1-R[t+1])X[t]}  X'[t] ] \notag \\ 
= 
r_2 (1 - e^{-\rho/\mu_1}) \mathbf{E} [  (\mathbf{1}_{X[t] > \rho} + \mathbf{1}_{X[t] < \rho,  X'[t] > X[t]} ) X'[t] ] \notag 
\end{align}
\begin{align}
= 
r_2 (1 - e^{-\rho/\mu_1}) (e^{-\rho / \mu_1} \mu_2 + \mathbf{E}[\mathbf{E}[ \mathbf{1}_{X[t] < \rho} \mathbf{1}_{X[t] < X'[t]} X'[t] | X'[t]] ]) \notag \\
= 
r_2 (1 - e^{-\rho/\mu_1}) (e^{-\rho / \mu_1} \mu_2 + \mathbf{E}[\mathbf{E}[ \mathbf{1}_{X[t] < \min(\rho, X'[t])}  X'[t] | X'[t]] ]) \notag \\
 = 
r_2 (1 - e^{-\rho/\mu_1}) (e^{-\rho / \mu_1} \mu_2 + \mathbf{E}[X'[t](1 - e^{-\min(\rho, X'[t])/\mu_1}) ]) \notag \\ 
 = 
r_2 (1 - e^{-\rho/\mu_1}) (e^{-\rho / \mu_1} \mu_2 + \mu_2 -  \mathbf{E}[X'[t] e^{-\min(\rho, X'[t])/\mu_1} ]) \notag \\ 
 = 
r_2 (1 - e^{-\rho/\mu_1}) (e^{-\rho / \mu_1} \mu_2 + \mu_2 - \int_0^\infty  x e^{-\min(\rho, x)/\mu_1} / \mu_2 e^{-x/\mu_2} dx ) \notag 
\end{align}
\begin{align}
 = 
r_2 (1 - e^{-\rho/\mu_1}) (e^{-\rho / \mu_1} \mu_2 + \mu_2 - \int_0^\rho  x  e^{- x/\mu_1} / \mu_2 e^{-x/\mu_2} dx - \int_\rho^\infty     x e^{- \rho /\mu_1} / \mu_2 e^{-x/\mu_2} dx ) \notag \\ 
 = 
r_2 (1 - e^{-\rho/\mu_1}) (e^{-\rho / \mu_1} \mu_2 + \mu_2 - \int_0^\rho    x/\mu_2 e^{- x(\mu_1 + \mu_2)/(\mu_1 \mu_2)} dx - e^{- \rho /\mu_1} / \mu_2 \int_\rho^\infty      x e^{-x/\mu_2} dx ) \notag \\ 
= 
r_2 (1 - e^{-\rho/\mu_1}) (e^{-\rho / \mu_1} \mu_2 + \mu_2 - [ - x \mu_1    /(\mu_1 + \mu_2) e^{- x(\mu_1 + \mu_2)/(\mu_1 \mu_2)} \notag \\ - \mu_1^2 \mu_2 / (\mu_1 + \mu_2)^2 e^{- x(\mu_1 + \mu_2)/(\mu_1 \mu_2)} ]_0^\rho 
 - e^{- \rho /\mu_1} / \mu_2 [ - x \mu_2 e^{-x/\mu_2} - \mu_2^2 e^{-x/\mu_2}]_\rho^\infty ) \notag \\ 
 = 
r_2 (1 - e^{-\rho/\mu_1}) (e^{-\rho / \mu_1} \mu_2 + \mu_2 + \rho \mu_1    /(\mu_1 + \mu_2) e^{- \rho (\mu_1 + \mu_2)/(\mu_1 \mu_2)} +  \frac{\mu_1^2 \mu_2}{(\mu_1 + \mu_2)^2} e^{-\rho(\mu_1 + \mu_2)/(\mu_1 \mu_2)} \notag \\
- \frac{\mu_1^2 \mu_2}{(\mu_1 + \mu_2)^2}
 - \rho e^{- \rho (\mu_1 + \mu_2) /(\mu_1 \mu_2)} - \mu_2 e^{-\rho(\mu_1 + \mu_2)/(\mu_1 \mu_2)} ) \\
\mathbf{E}[ \mathbf{1}_{R[t]=0} \mathbf{1}_{(1-r_2) X'[t] + (1-r_2) Y[t] > (1-R[t+1])(1 - r_2)X[t] + (1-r_2)Y[t]} r_2 Y[t]] \notag \\
= 
\mathbf{E}[ \mathbf{1}_{R[t]=0} \mathbf{1}_{ X'[t] > (1-R[t+1])X[t]} r_2 Y[t]] 
= 
r_2 \mu_3 (1 - e^{-\rho/\mu_1}) \mathbf{E}[  \mathbf{1}_{ X'[t] > (1-R[t+1])X[t]}] \notag \\ 
= 
r_2 \mu_3 (1 - e^{-\rho/\mu_1}) \mathbf{E}[ \mathbf{E}[ \mathbf{1}_{ X'[t] > (1-R[t+1])X[t]} | X'[t]] ] \notag \\
= 
r_2 \mu_3 (1 - e^{-\rho/\mu_1}) \mathbf{E}[ \mathbf{E}[ \mathbf{1}_{ X[t] > \rho} + \mathbf{1}_{X[t] < \rho, X'[t] > X[t]} | X'[t]] ] \notag \\
= 
r_2 \mu_3 (1 - e^{-\rho/\mu_1}) \mathbf{E}[ e^{-\rho/\mu_1} +  \mathbf{E}[ \mathbf{1}_{X[t] < \min( \rho, X'[t])} | X'[t]] ] \notag \\
= 
r_2 \mu_3 (1 - e^{-\rho/\mu_1}) \mathbf{E}[ e^{-\rho/\mu_1} +  1 - e^{-\min( \rho, X'[t])/\mu_1}  ] \notag \\
= 
r_2 \mu_3 (1 - e^{-\rho/\mu_1}) ( e^{-\rho/\mu_1} +  1 - \int_0^\infty e^{-\min( \rho, x)/\mu_1}/\mu_2 e^{-x/\mu_2} dx  ) \notag \\
= 
r_2 \mu_3 (1 - e^{-\rho/\mu_1}) ( e^{-\rho/\mu_1} +  1 - \int_0^\rho e^{- x/\mu_1}/\mu_2 e^{-x/\mu_2} dx  - \int_\rho^\infty e^{-\rho/\mu_1}/\mu_2 e^{-x/\mu_2} dx ) \notag \\
= 
r_2 \mu_3 (1 - e^{-\rho/\mu_1}) ( e^{-\rho/\mu_1} +  1 - [-\mu_1/(\mu_1 + \mu_2) e^{-x(\mu_1 + \mu_2)/(\mu_1 \mu_2)}]_0^\rho   -  e^{-\rho/\mu_1} e^{-\rho/\mu_2} )  \notag \\
= 
r_2 \mu_3 (1 - e^{-\rho/\mu_1}) ( e^{-\rho/\mu_1} +  1 + \mu_1/(\mu_1 + \mu_2) e^{-\rho (\mu_1 + \mu_2)/(\mu_1 \mu_2)} - \mu_1/(\mu_1 + \mu_2)   -  e^{-\rho/\mu_1} e^{-\rho/\mu_2}  )  . \label{eq:vsec1}
\end{align}
\end{proof}

\section{Proof of Lemma~\ref{lem:validatorpolicymucondi}}
\label{apx:validatorpolicymucondi}
\begin{proof}
Setting $r_1 = 0$ and $\mu_1 = 1 - \mu_2$, we simplify Equations~\eqref{eq:validatorrewpolicy} and~\eqref{eq:validrewdeviation} to get\footnote{All simplifications in this section were performed using Wolfram Mathematica [XXX].} 
\begin{align}
\mathbf{E}[V_p^\mathrm{policy}[t]] - \mathbf{E}[V_p^\mathrm{default}[t]] = e^{\frac{(1+\mu_2)\rho}{(-1+\mu_2)\mu_2}}\left( e^{\frac{\rho}{\mu_2}} \mu_2 (-1+r_2) - \mu_2^2 (-1 + r_2) \right. \notag \\
\left. + e^{\frac{\rho}{\mu_2(1-\mu_2))}}(1-\mu_2 +\mu_2^2 (-1 + r_2) + \rho) \right).
\end{align}
Now, 
\begin{align}
e^{\frac{\rho}{\mu_2}} \mu_2 (-1+r_2) - \mu_2^2 (-1 + r_2) 
 + e^{\frac{\rho}{\mu_2(1-\mu_2))}}(1-\mu_2 +\mu_2^2 (-1 + r_2) + \rho) \notag \\
 > -e^{\frac{\rho}{\mu_2}}\mu_2 + e^{\frac{\rho}{\mu_2}} (1 - \mu_2 + \mu_2^2(-1 + r_2) + \rho) \notag \\
 > -e^{\frac{\rho}{\mu_2}} \mu_2 + e^{\frac{\rho}{\mu_2}} (1-\mu_2 - \mu_2 + \rho ) \notag \\
 > -e^{\frac{\rho}{\mu_2}} \mu_2 + e^{\frac{\rho}{\mu_2}} \rho \notag \\
> 0, 
\end{align}
where the last two inequalities hold as long as $\mu_2 < 1/2$ and $\rho > \mu_2$. 
\end{proof}

\section{Proof of Lemma~\ref{lem:fixedptlesszero}}
\label{apx:fixedptlesszero}
\begin{proof}
Simplifying and setting $r_1 = 0, \mu_2 = 1/2$, we have 
\begin{align}
\mathbf{E}[V_\mathrm{primary}[t]]\mu_2 - \mathbf{E}[V_\mathrm{secondary}[t]] \mu_1 = r_2( (-0.25 -0.5\mu_3 -0.5\rho) e^{-6\rho} + (1.5\mu_3 +  0.5 + 0.5\rho)e^{-4\rho} \notag \\
+ (-0.25-0.5\mu_3 -0.5\rho) e^{-2\rho})  \notag \\
< r_2( (1.5\mu_3 + 0.5 + 0.5\rho) e^{-4\rho } + (-0.25 -0.5\mu_3 -0.5\rho)e^{2\rho} ) \notag \\
= r_2( (-0.25 +0.5 e^{-2\rho} + \mu_3(-0.5 + 1.5 e^{-2\rho}) + \rho (-0.5 + 0.5 e^{-2\rho})) e^{-2\rho}) < 0,
\end{align}
as long as $0.5 e^{-2\rho}<0.25$, $1.5e^{-2\rho}<0.5$ and $e^{-2\rho} < 1$, or as long as $\rho > \ln(3)/2$, thus completing the proof. 
\end{proof}

\section{Proof of Theorem~\ref{thm:fixedpointlesshalf}}
\label{apx:fixedpointlesshalf}
\begin{proof}
For any $\mu_2 \in (0, 0.5), r_1 = 0$ we have 
\begin{align}
(\mathbf{E}[V_\mathrm{primary}[t]]\mu_2 - \mathbf{E}[V_\mathrm{secondary}[t]] \mu_1)/r_2 \notag \\
= -3\mu_2^2 + 2 \mu_2^3 + \mu_2 + e^{\frac{\rho}{-1+\mu_2}} (4 \mu_2^2 -2 \mu_2^3 - \mu_3 + \mu_2 (-2 + \mu_3 - \rho)) + e^{\frac{2\rho}{-1+\mu_2}} (-\mu_2^2 + \mu_3 + \mu_2) \notag \\
e^{\frac{\rho}{(-1+\mu_2) \mu_2}} (-2 \mu_2^2 -2 \mu_2^3 + \mu_2 (\mu_3 + \rho)) + e^{\frac{(1+\mu_2)\rho}{-1+\mu_2}\mu_2} (-2 \mu_2^2 -2\mu_2^3 -\mu_2 (\mu_3 + \rho)) \notag \\
> -3 \mu_2^2 + 2 \mu_2^3 + \mu_2 + e^{\frac{-\rho}{0.5}} (-2*0.5^3 -\mu_3 + 0.5 (-2 + \mu_3 - \rho)) + e^{\frac{-2\rho}{0.5}} (-0.25 + \mu_3) \notag \\
+ e^{\frac{-\rho}{0.25}} (-0.5 -0.25) + e^{\frac{-\rho}{0.25}} (-0.5 -0.25 -0.5 (\mu_3 + \rho)),
\end{align}
for $\rho > \mu_3 - 2$. 
For any $\mu_2 \in (0, 0.5)$ let $c := -3 \mu_2^2 + 2 \mu_2^3 + \mu_2 $. 
Note that $c > 0$. 
In the following we show that each of the four terms in the above summation (following $-3\mu_2^2 + 2 \mu_2^3 + \mu_2$ are greater than $-c/4$. 

First, we want 
\begin{align}
e^{\frac{-\rho}{0.5}} (-2 * 0.5^3 -\mu3 + 0.5 (-2 + \mu_3 -\rho)) > \frac{-c}{4}, 
\end{align}
or equivalently 
\begin{align}
e^{\frac{-\rho}{0.5}} (1.25 + 0.5 \mu_3 + 0.5 \rho ) < \frac{c}{4}.  
\end{align}
If $\rho > 0.5 (1.25 + 0.5\mu_3) - 0.5 \ln(\frac{ce}{8})$ we have 
\begin{align}
e^{\frac{-\rho}{0.5}} (1.25 + 0.5 \mu_3 ) < e^{-(1.25 + 0.5 \mu_3) + \ln(ce)/8} (1.25 + 0.5 \mu_3) < \frac{1}{e} \frac{ce}{8} = \frac{c}{8}.
\end{align}
Since $1.25 + 0.5 \mu_3 - \ln(ce/8) > 1$ and $e^{-x} x $ is a decreasing function of $x$ for $x > 1$, we have 
\begin{align}
e^{\frac{-\rho}{0.5}} (0.5 \rho) < e^{-(1.25 + 0.5 \mu_3 ) + \ln(ce/8)} (0.5 (0.5 (1.25 + 0.5 \mu_3 ) - 0.5 \ln(ce/8))) \notag \\
< e^{-(1.25 + 0.5\mu_3) + \ln(ce/8)} (0.25 (1.25 + 0.5 \mu_3) - 0.25 \ln(ce/8)) \notag \\
< e^{-(1.25 + 0.5 \mu_3) + \ln(ce/8)} (0.25 (1.25 + 0.5 \mu_3)) \notag \\
< e^{\ln(ce/8)} 0.25 / e = \frac{ce 0.25 }{8e} < \frac{c}{8}.
\end{align}
Therefore, $e^{\frac{-\rho}{0.5}} (1.25 + 0.5\mu_3 + 0.5\rho ) < c/4$. 

Next, we want $e^{\frac{-2\rho}{0.5}} (-0.25 + \mu_3) > -c/4$, or equivalently $e^{\frac{-2\rho}{0.5}} (0.25 - \mu_3) < c/4$. 
If $\rho > 0.25 (0.25 - \mu_3) - 0.25\ln(ce/4)$ we have 
\begin{align}
e^{\frac{-2\rho}{0.5}} (0.25 - \mu_3) < e^{-(0.25-\mu_3) + \ln(ce/4)} (0.25 - \mu_3) \notag \\
< e^{\ln(ce/4)} / e = \frac{ce}{4e} = \frac{c}{4}. 
\end{align}

Next, we want $e^{\frac{-\rho}{0.25}} (-0.5 -0.25) > -c/4$, or equivalently $e^{\frac{-\rho}{0.25}} (0.75) < c/4$. 
If $\rho > -0.25 \ln(c/3)$ we get what we want. 

Next, we want $e^{\frac{-\rho}{0.25}} (-0.5 -0.25 -0.5 (\mu_3 + \rho)) > -c/4$ or equivalently $e^{\frac{-\rho}{0.25}} (0.75 + 0.5 (\mu_3 + \rho)) < c/4$.
If $\rho > 0.25 (0.75 + 0.5 \mu_3) - 0.25\ln(ce/8)$, we have 
\begin{align}
e^{\frac{-\rho}{0.25}} (0.75 + 0.5 \mu_3) &< e^{-(0.75 + 0.5\mu_3) + \ln(ce/8)} (0.75 + 0.5 \mu_3) \notag \\
&< e^{\ln(ce/8)}/e = \frac{c}{8}.
\end{align}

Note that $c$ is at most 0.1 (occuring at the maximum of the polynomial $-3x^2 + 2x^3 + x$ for $x\in (0, 0.5)$). 
Therefore, $0.75 + 0.5 \mu_3 - \ln(ce/8) > 1$.
We have 
\begin{align}
e^{\frac{-\rho}{0.25}} 0.5 \rho &< e^{-(0.75 + 0.5 \mu_3) + \ln(ce/8)} 0.5 (0.25 (0.75 + 0.5 \mu_3 ) - 0.25 \ln(ce/8)) \notag \\
&< e^{-(0.75 + 0.5\mu_3) + \ln(ce/8)} 0.5 (0.25 (0.75 + 0.5 \mu_3)) \notag \\
&< e^{\ln(ce/8)} 0.5 * 0.25/e = ce 0.5 * 0.25 / (8e) < c/8.  
\end{align}

Hence for any $0 <\mu_2 < 0.5$ and $c = -3 \mu_2^2 + 2 \mu_2^3 + \mu_2$, if $\rho > \max(\mu_3 - 2, 0.5 (1.25 + 0.5 \mu_3) - 0.5 \ln(ce/8), 0.25 (0.25 - \mu_3) - 0.25 \ln(ce/4), - 0.25 \ln(c/3), 0.25 (0.75 +0.5 \mu_3 ) - 0.25 \ln(ce/8))$ we have 
\begin{align}
-3\mu_2^2 + 2 \mu_2^3 + \mu_2 + e^{\frac{\rho}{-1+\mu_2}} (4 \mu_2^2 -2 \mu_2^3 - \mu_3 + \mu_2 (-2 + \mu_3 - \rho)) + e^{\frac{2\rho}{-1+\mu_2}} (-\mu_2^2 + \mu_3 + \mu_2) \notag \\
e^{\frac{\rho}{(-1+\mu_2) \mu_2}} (-2 \mu_2^2 -2 \mu_2^3 + \mu_2 (\mu_3 + \rho)) + e^{\frac{(1+\mu_2)\rho}{-1+\mu_2}\mu_2} (-2 \mu_2^2 -2\mu_2^3 -\mu_2 (\mu_3 + \rho)) > 0.
\end{align}
From Lemma~\ref{lem:fixedptlesszero} and by the continuity of the function above in $\mu_2$, there must exist a fixed point between $\mu_2$ and 0.5. 
\end{proof}

\end{document}